%% file: main.tex
\documentclass[journal,twoside,web]{ieeecolor}
\usepackage{generic}
\usepackage{cite}
\usepackage[markup=underlined,commentmarkup=footnote]{changes}
\usepackage[colorlinks,urlcolor=blue,linkcolor=blue,citecolor=blue]{hyperref}
\usepackage{textcomp}

\usepackage{times}
\usepackage{tablefootnote}
\usepackage{multicol}
\usepackage{amsthm}
\usepackage{amsmath} %
\usepackage{amssymb}  %
\usepackage{amsfonts}
\usepackage{mathtools}
\usepackage{version,xspace}
\usepackage{url,doi}
\usepackage{multirow}

 \newcommand{\setdef}[2]{\{#1
	\; | \; #2\}}

\usepackage{algorithm}
\usepackage{algpseudocode}

\usepackage{xcolor}
\newcommand\oprocendsymbol{\hbox{$\triangle$}}
\newcommand\oprocend{\relax\ifmmode\else\unskip\hfill\fi\oprocendsymbol}

\let\labelindent\relax
\usepackage{enumitem} \renewcommand{\theenumi}{(\roman{enumi})}

\DeclareSymbolFont{bbold}{U}{bbold}{m}{n}
\DeclareSymbolFontAlphabet{\mathbbold}{bbold}
\newcommand{\vect}[1]{\mathbbold{#1}}

\newcommand{\real}{\mathbb{R}}

\newcommand{\change}[1]{{#1}}

\newcommand{\seminorm}[1]{{\left\vert\kern-0.25ex\left\vert\kern-0.25ex\left\vert #1
		\right\vert\kern-0.25ex\right\vert\kern-0.25ex\right\vert}}

\newcommand{\semimeasure}[1]{\mu_{\seminorm{\cdot}}\kern-0.5ex\left(#1\right)}

\newcommand{\suchthat}{\;\ifnum\currentgrouptype=16 \middle\fi|\;}
\newcommand{\scirc}{\raise1pt\hbox{$\,\scriptstyle\circ\,$}}

\newcommand{\OF}{\mathsf{F}}
\newcommand{\OG}{\mathsf{G}}

\newcommand{\ON}{\mathsf{N}}

\input{macros}

\newtheorem{theorem}{Theorem}
\newtheorem{proposition}{Proposition}
\newtheorem{corollary}{Corollary}

\newtheorem*{lemma*}{Lemma}
\theoremstyle{definition}
	\newtheorem{definition}{Definition}
	\newtheorem{remark}{Remark}
	\newtheorem{example}{Example}

	\newtheorem*{example*}{Example}
\newtheorem*{tap}{Target-Avoid Problem}

\def\BibTeX{{\rm B\kern-.05em{\sc i\kern-.025em b}\kern-.08em
    T\kern-.1667em\lower.7ex\hbox{E}\kern-.125emX}}
\begin{document}
\title{Efficient Interaction-Aware Interval Analysis of Neural Network Feedback Loops}
\author{Saber Jafarpour$^*$, \IEEEmembership{Member, IEEE}, Akash Harapanahalli$^*$,  \IEEEmembership{Graduate Student Member, IEEE}, and Samuel Coogan, \IEEEmembership{Senior Member, IEEE}
\thanks{$^*$ These authors contributed equally.}
\thanks{This work was supported in part by the National Science Foundation under grants 1749357 and 2219755 and the Air Force Office of Scientific Research under Grant FA9550-23-1-0303.}
\thanks{Saber Jafarpour is with Department of Electrical, Computer, and Energy Engineering, University of Colorado Boulder, Boulder, Colorado, USA (e-mail: saber.jafarpour@colorado.edu\})}
\thanks{Akash Harapanahalli and Samuel Coogan are with School of Electrical and Computer Engineering, Georgia Institute of Technology, Atlanta, GA, USA (e-mail: \{aharapan, ~sam.coogan\}@gatech.edu\})}}

\maketitle

\begin{abstract}
In this paper, we propose a computationally efficient framework for interval reachability of systems with neural network controllers. 
Our approach leverages inclusion functions for the open-loop system and the neural network controller to embed the closed-loop system into a larger-dimensional embedding system, where a single trajectory over-approximates the original system's behavior under uncertainty.
We propose two methods for constructing closed-loop embedding systems, which account for the interactions between the system and the controller in different ways. 
The interconnection-based approach considers the worst-case evolution of each coordinate separately by substituting the neural network inclusion function into the open-loop inclusion function. 
The interaction-based approach uses novel Jacobian-based inclusion functions to capture the first-order interactions between the open-loop system and the controller by leveraging state-of-the-art neural network verifiers.
Finally, we implement our approach in a Python framework called \texttt{ReachMM} to demonstrate its efficiency and scalability on benchmarks and examples ranging to $200$ state dimensions.

\end{abstract}


\maketitle

\section{Introduction}

\paragraph*{Problem Description and Motivations}

Recent advances in machine learning are bringing learning algorithms into the heart of safety-critical control systems, such as autonomous vehicles, manufacturing sectors, and robotics. 
For control systems, these learning algorithms often act in the closed-loop setting, as direct feedback controllers~\cite{TZ-GK-SL-PA:16} or motion planners~\cite{AHQ-AS-MJB-MCY:19}.
Learning-based closed-loop controllers can improve the performance of systems while reducing their computational burden for online implementations as compared with more traditional optimization-based approaches. 
However, despite their impressive performance, learning models are known to be overly sensitive to input disturbances~\cite{CZ-WZ-IS-JB-DE-IG-RF:13}: a small input perturbation can lead to comparatively large changes in their output. This issue amplifies in feedback settings, as disturbances can accumulate in the closed-loop. 
As a result, ensuring reliability of learning algorithms is an essential challenge in their integration into safety-critical systems.

Typical learning architectures involve high-dimensional nonlinear
function approximators, such as neural networks, necessitating special
tools for their input-output analysis.
The machine learning and control community have made significant
progress in analyzing the safety of neural networks in isolation,
including efficient and sound input-output bounds, worst-case
adversarial guarantees, and sampling-based stochastic
guarantees~(cf.~\cite{CL-TA-CL-CS-CB-MJK:21}).
However, most of these existing frameworks for standalone neural
networks do not address the unique challenges for closed-loop
analysis---namely information propagation, non-stationarity of the
bounds, and complex interactions between the system and the
controller.
In these cases, it is essential to understand and capture the nature
of the interactions between the system and the neural network.
Recently, several frameworks have emerged for safety verification of
learning algorithms in the closed-loop.
These frameworks usually incorporate neural network verification
algorithms into the closed-loop safety analysis by studying their
interactions with the open-loop system.
However, most of these existing frameworks are either limited to
a specific class of systems or learning algorithms, or they impose
significant computational burdens, rendering them unsuitable for
runtime safety verification.

\paragraph*{Literature Review}

Safety verification of nonlinear dynamical systems without learning-enabled systems is a classical problem and is typically solved using reachability analysis tools, although significant challenges remains even in this setting. Reachability of nonlinear systems has been studied using optimization-based methods such as the Hamilton-Jacobi approach~\cite{SB-MC-SH-CJT:17} and the level set approach~\cite{IM-CJT:00}. Several computationally tractable approaches including the ellipsoidal method~\cite{ABK-PV:00} and the zonotope method~\cite{AG:05} have been developed for reachability analysis of dynamical systems. We refer to~\cite{XC-SS:22} for a recent review of state-of-the-art reachability analysis techniques. Interval analysis is a classical framework for computing interval functional bounds~\cite{LJ-MK-OD-EW:01} which has been successfully used for reachability analysis of dynamical systems~\cite{MA-OS-MB:07,JKS-PIB:13,KS-JKS:17}. A closely-related body of literature is the mixed monotone theory, which extends the classical monotone system theory by separating the cooperative and competitive effect of the dynamics~\cite{GAE-HLS-EDS:06,HLS:08}. Recently, mixed monotone theory has been used as an efficient framework for reachability analysis of dynamical systems~\cite{SC-MA:15b,SC:20,MA-MD-SC:21,MK-SZY:23}.   

Rigorous verification of standalone neural networks has been studied using abstract interpretation techniques such as Reluplex~\cite{GK-CB-DLD-KJ-MJK:17} and Neurify~\cite{SW-KP-JW-JY-SJ:18}, using interval bound propagation methods~\cite{MM-TG-MV:18,SG-etal:19}, and convex-relaxation approaches such as LipSDP~\cite{MF-MM-GJP:22} and  CROWN~\cite{HZ-etal:18} and its variants~\cite{SW-HZ-KX-XL-SJ-CJH-ZJK:21}.
The simplest method for studying reachability of neural network controlled systems is via an \emph{input-output approach}:
using existing techniques for estimating the output of the neural network, the input-output approach performs reachability analysis of the closed-loop system by substituting the full output range of the neural network as the input set of the open-loop system. Examples of this approach include NNV~\cite{HDT-etal:20}, and simulation-guided interval analysis~\cite{WX-HDT-XY-TTJ:21}. 
This approach is generally computationally efficient and applicable to general forms of neural networks and control systems, but it can lead to overly-conservative estimates of reachable sets~\cite[Section 2.1]{SD-XC-SS:19}. The reason for this over-conservatism is that this approach essentially
treats neural networks controllers as adversaries and, thus, cannot capture the beneficial interactions between the controller and the system. An alternative framework is the \emph{functional approach}, which is based on composing function approximations of the system and the neural network controller.
Examples of this approach for linear systems include using linear programming in ReachLP~\cite{ME-GH-CS-JPH:21}, using semi-definite programming in Reach-SDP~\cite{HH-MF-MM-GJP:20}, and using mixed integer programming in~\cite{CS-AM-AI-MJK:22}.
For nonlinear systems, the functional approach has been used in ReachNN~\cite{CH-JF-WL-XC-QZ:19}, Sherlock~\cite{SD-XC-SS:19}, Verisig 2.0~\cite{RI-TC-JW-RA-GP-IL:21}, POLAR~\cite{CH-JF-XC-WL-QZ:22} and JuliaReach~\cite{SB-etal:19,CS-MF-SG:22}. Functional methods are able to capture beneficial interactions between the neural network controller and the system, but they are often computationally burdensome and generally do not scale well to large-scale systems.

\paragraph*{Contributions}

   In this paper, we introduce a general and computationally efficient framework for studying reachability of continuous-time closed-loop systems with nonlinear dynamics and neural networks controllers. First, we review the notion of inclusion function from classical interval analysis to over-approximate the input-output behavior of functions. 
       Our first minor result states that if an inclusion function is chosen to be minimal, then it provides tighter over-approximations of the output behavior than a Lipschitz bound of the function. 
       For a given function, we study different approaches for constructing inclusion functions. In particular, we propose a novel Jacobian-based cornered inclusion function, which is specially amenable to analysis of closed-loop systems. 
       We further show that this class of inclusion functions plays a critical role in integrating the existing off-the-shelf neural network verification algorithms into our framework. 

       Second, using the notion of inclusion functions, we develop a general approach for interval analysis of continuous-time dynamical systems.
       The key idea is to use the inclusion functions to embed the original dynamical system into a higher dimensional embedding system.
       Then trajectories of the embedding system can be used to provide hyper-rectangular over-approximation for reachable sets of the original system.
       Our approach unifies and extends several existing approaches for interval analysis of dynamical systems. Notably, our novel proof technique, based on the classical monotone comparison Lemma, allows for accuracy comparison between different inclusion functions.

       As the culminating contribution of this paper, we develop a computationally efficient framework for reachability analysis of nonlinear systems with neural network controllers. 
       The key in our framework is a careful combination of the inclusion function for the open-loop system and the inclusion function for the neural network controller to obtain an inclusion function for the closed-loop system.   
    We propose two methods for combining the open-loop inclusion function and the neural network inclusion function. 
    The second method, called the interaction-based approach, constructs a closed-loop embedding system using a Jacobian-based inclusion function for the open-loop system and local affine bounds for the neural network.  %
    Compared to the interconnection-based approach, the interaction-based method completely captures the first order interactions between the system and the neural network. %

    Finally, we implement our approach in a Python toolbox called \texttt{ReachMM} and perform several numerical experiments to show the efficiency and accuracy of our framework. For several linear and nonlinear benchmarks, we show that our method beats the state-of-the-art methods. %
    Moreover, we study scalability of our approach and compare it with the state-of-the-art methods on a platoon of vehicles with up to 200 state variables. Compared to our conference paper~\cite{SJ-AH-SC:23}, this paper provides a more general framework for interval reachability using inclusion functions, introduces the interaction-based approach to capture the first order interactions between the system and the neural network controller, and includes a suite of numerical experiments. 
    We would like to highlight that, in this paper, we focus on continuous-time dynamical systems. However, with minor modifications, our framework can be used to analyze reachability of discrete-time systems. 

\section{Mathematical Preliminary and Notations}

For every two vectors $v,w\in \real^n$ and every
subset $S\subseteq \{1,\ldots,n\}$, we define the vector $v_{[S:w]}\in \real^n$ by $ \left(v_{[S:w]}\right)_j = \begin{cases}
    v_j & j\not\in S\\
    w_j & j\in S.
    \end{cases}$.
    Given a matrix
$B \in \mathbb{R}^{n\times m}$, we denote the non-negative part of $B$
by $[B]^+ = \max(B, 0)$ and the nonpositive part of $B$ by
$[B]^- = \min(B, 0)$. Given a square matrix $A\in \real^{n\times n}$ the Metzler and non-Metzler part of $A$
are denoted by $[A]^{\mathrm{M}}$ and
$[A]^{\mathrm{nM}}$, respectively, where $
  ([A]^{\mathrm{M}})_{ij} =\begin{cases}
    A_{ij} & A_{ij} \geq 0\; \mbox{or} \; i =  j\\
    0 & \mbox{otherwise,}
  \end{cases}$ and $[A]^{\mathrm{nM}}= A-[A]^{\mathrm{M}}$. 
Given a map $g:\real^n\to \real^m$, the $\ell_{\infty}$-norm Lipschitz bound of $g$ is the smallest $L\in \real_{\ge 0}$ such that 
\begin{align*}
    \|g(x)-g(y)\|_{\infty}\le L\|x-y\|_{\infty},\qquad \mbox{ for all }x,y\in \real^n.
\end{align*}
We denote the $\ell_{\infty}$-norm Lipschitz constant of $g$ by $\mathrm{Lip}_{\infty}(g)$. For a differentiable map $g:\real^n\to \real^m$, we denote its Jacobian at point $x\in \real^n$ by $Dg(x)$. Consider the dynamical system
\begin{align}  
    \dot{x}&=f(x,w), &&\mbox{for all } t\in \real_{\ge 0},\label{eq:dynamics-c}
\end{align}
where $x\in \real^n$ is the state of the system and $w\in \real^q$ is the disturbance. Given a piecewise continuous disturbance $t\mapsto w(t)$ and an initial condition $x_0\in \real^n$, the trajectory of~\eqref{eq:dynamics-c} starting from $x_0$ is denoted by $\phi_f(t,x_0,w(t))$. Let $\mathcal{X}\subseteq \real^n$ and $\mathcal{W}\subseteq \real^q$, then the reachable set of~\eqref{eq:dynamics-c} starting from the initial set $\mathcal{X}$ and with the disturbance set $\mathcal{W}$ is defined by:  
\begin{equation}
    \mathcal{R}_{f}(t,\mathcal{X},\mathcal{W}) = \left\{
    \begin{aligned}
    \phi&_{f}(t, t_0, x_0, w),\,\forall x_0\in\mathcal{X},\\& w:\real\to\mathcal{W} \text{ piecewise cont.}
    \end{aligned}
    \right\}
\end{equation} 
The set $\mathcal{X}$ is \emph{robustly forward invariant} if  $\mathcal{R}_{f}(t,\mathcal{X},\mathcal{W})\subseteq \mathcal{X}$ for all $t\geq 0$. The system~\eqref{eq:dynamics-c} is continuous-time monotone on $\mathcal{X}\subseteq \real^n$ if, for every $i\in\{1,\ldots,n\}$, every $x\le y\in \mathcal{X}$ with $x_i=y_i$, and every $u\le v$, we have
    $f_i(x,u) \le f_i(y,v).$ 
If a dynamical system~\eqref{eq:dynamics-c} is monotone on $\mathcal{X}$ and $\mathcal{X}$ is robustly forward invariant, then its trajectories preserve the standard partial order by time, i.e., for every two trajectories $t\mapsto x(t)$ and $t\mapsto y(t)$ of the systems~\eqref{eq:dynamics-c}, if $x(0)\le y(0)$, then $x(t)\le y(t)$ for all $t\in \real_{\ge 0}$~\cite{HLS:95}. 
The \emph{southeast} partial order $\le_{\mathrm{SE}}$ on $\real^{2n}$ is defined by $\left[\begin{smallmatrix}
    x\\
    \widehat{x}
    \end{smallmatrix}\right]\le_{\mathrm{SE}} \left[\begin{smallmatrix}
    y\\
    \widehat{y}
    \end{smallmatrix}\right]$ if and only if $x\le y$ and $\widehat{y}\le \widehat{x}$. %
We also define $\mathcal{T}^{2n}_{\ge 0}=\setdef{\left[\begin{smallmatrix}
    x\\
    \widehat{x}
    \end{smallmatrix}\right]\in \real^{2n}}{x\le \widehat{x}}$ and $\mathcal{T}^{2n}_{\le 0}=\setdef{\left[\begin{smallmatrix}
    x\\
    \widehat{x}
    \end{smallmatrix}\right]\in \real^{2n}}{x\ge \widehat{x}}$ and $\mathcal{T}^{2n}=\mathcal{T}^{2n}_{\ge 0}\bigcup \mathcal{T}^{2n}_{\le 0}$.
\section{Problem Statement}
We consider a system described by:
\begin{align} 
    \dot{x} &= f(x,u,w),&& t\in \real_{\ge 0},\label{eq:plant-c}
\end{align}
where $x\in \real^n$ is the state of the system, $u\in \real^p$ is the control input, and $w\in \real^q$ is the disturbance.  We assume that $f:\real^n\times \real^p\times \real^q\to \real^n$ is a parameterized vector field and the state feedback is parameterized by a $k$-layer
feed-forward neural network controller $N:\real^{n}\to \real^p$ defined by:
  \begin{align}\label{eq:NN}
  \xi^{(0)}&= x\nonumber\\
   \xi^{(i)} &= \phi^{(i-1)}(W^{(i-1)} \xi^{(i-1)}+ b^{(i-1)}), \quad i\in \{1,\ldots,k\}\nonumber\\
       u &= W^{(k)}\xi^{(k)}(x) +b^{(k)} := N(x),
        \end{align}
 where $n_i$ is the number of neurons in the $i$th layer, $W^{(i-1)}\in \real^{n_i\times n_{i-1}}$ is the weight matrix of the $i$th layer, $b^{(i-1)}\in \real^{n_i}$ is the bias vector of the $i$th layer, $\xi^{(i)}(y)\in \real^{n_i}$ is the $i$th layer hidden variable, and
$\phi^{(i-1)}:\real^{n_{i}}\to \real^{n_i}$ is the $i$th layer diagonal activation function satisfying $0\le \frac{\phi_j^{(i-1)}(x)-\phi_j^{(i-1)}(y)}{x-y}\le 1$
for every $j\in \{1,\ldots,n_i\}$. One can show that a large class of activation functions including ReLU, leaky ReLU, sigmoid, and tanh satisfies this condition (after a possible re-scaling of their co-domain). With this neural network feedback controller, the closed-loop system is given by:
 \begin{align}
     \dot{x} &= f(x(t),N(x(t)),w) := f^{\mathrm{c}}(x(t),w).\label{eq:closedloop-c}
 \end{align}
We assume that $\mathcal{X}_0\subseteq \real^n$ is the initial set of states for the closed-loop system~\eqref{eq:closedloop-c}. This set can represent the uncertainties in the starting state of the system or localization errors in estimating the current state of the system. Moreover, we assume that the disturbance $w$ belongs to the set $\mathcal{W}\subseteq \real^n$ representing the model uncertainty or exogenous disturbances on the system. In this paper, we focus on the target-avoid problem as defined below:

\begin{tap}
Given a target set $\mathcal{G}\subset \real^n$, and an unsafe set $\mathcal{S}_{\mathrm{unsafe}}\subset\real^n$ and a time interval $[0,T]$, check if the closed-loop system~\eqref{eq:closedloop-c} reaches the target $\mathcal{G}$ at time $T$ while avoiding the unsafe region $\mathcal{S}_{\mathrm{unsafe}}$. Mathematically, 
\begin{enumerate}
    \item $\mathcal{R}_{f^{\mathrm{c}}}(T,\mathcal{X}_0,\mathcal{W})\subseteq \mathcal{G}$, 
    \item $\mathcal{R}_{f^{\mathrm{c}}}(t,\mathcal{X}_0,\mathcal{W})\bigcap \mathcal{S}_{\mathrm{unsafe}}=\emptyset$, for all $t\in [0,T]$. 
\end{enumerate}
\end{tap}
The target-avoid problem is a classical and prevalent notion of safety in many engineering applications, especially those involving safety-critical systems~\cite{XC-SS:22}. Moreover, diverse objectives including multiagent coordination, complex planning specified with formal languages and logics, and classical control-theoretic criteria such as robust invariance and stability can be achieved by concatenating and combining different instantiations of the target-avoid problem. 

In general, computing the exact reachable sets of the closed-loop system~\eqref{eq:closedloop-c} %
is not computationally tractable. Our approach to solve the target-avoid safety problem is based on constructing a computationally efficient over-approximation $\overline{\mathcal{R}}_{f^c}(t,\mathcal{X}_0,\mathcal{W})$ of the reachable set of the closed-loop system. Then, reaching the target is guaranteed by $\overline{\mathcal{R}}_{f^c}(T,\mathcal{X}_0,\mathcal{W}) \subseteq \mathcal{G}$ and  avoiding the unsafe region is guaranteed when $\overline{\mathcal{R}}_{f^c}(t,\mathcal{X}_0,\mathcal{W})\bigcap \mathcal{S}_{\mathrm{unsafe}}=\emptyset$, for every $t\in [0,T]$.

 \section{Interval Analysis and Inclusion Functions}\label{sec:IA}

In this section, we develop a theoretical framework for interval reachability analysis of arbitrary mappings. The key element in our framework is the notion of inclusion function from interval arithmetic~\cite{LJ-MK-OD-EW:01,REM-RBK-MJC:09}. 

\subsection{Inclusion Functions}

Given a nonlinear input-output map $g:\real^n\to \real^m$ and an interval input, the inclusion function of $g$ provides interval bounds on the output of $g$.

\begin{definition}[Inclusion function~\cite{LJ-MK-OD-EW:01}]\label{def:embedding}
Given a map $g:\real^n\to \real^m$, the function $\OG = \left[\begin{smallmatrix}\underline{\OG}\\ \overline{\OG}\end{smallmatrix}\right]:\mathcal{T}_{\ge 0}^{2n}\to \mathcal{T}_{\ge 0}^{2m}$ is
\begin{enumerate}
    \item an \emph{inclusion function} for $g$ if, for every $x\le \widehat{x}$ and every  $z\in [{x},\widehat{x}]$, we have
\begin{align}
\label{eq:incdef}
     \underline{\OG}({x},\widehat{x})\le g(z)\le 
         \overline{\OG}({x},\widehat{x}).
\end{align}
\item a \emph{$[y,\widehat{y}]$-localized inclusion function} if \eqref{eq:incdef} holds for every $ [{x},\widehat{x}]\subseteq [y,\widehat{y}]$ and every  $z\in [{x},\widehat{x}]$.
\end{enumerate}
Moreover, an inclusion function $\OG$ for $g$ is 
\begin{enumerate}\setcounter{enumi}{2}
    \item \emph{monotone} if, for every $\left[\begin{smallmatrix}x\\ \widehat{x}\end{smallmatrix}\right] \le_{\mathrm{SE}} \left[\begin{smallmatrix}y\\ \widehat{y}\end{smallmatrix}\right]$, we have $\OG({x},\widehat{x})\le_{\mathrm{SE}}\OG({y},\widehat{y})$;
    \smallskip
    \item \emph{thin} if, for every $x\in \real^n$, $\underline{\OG}(x,x)=\overline{\OG}(x,x)=g(x)$;  
    \smallskip \item\label{p3:minimal}
    \emph{minimal} if, for every ${x}\le \widehat{x}$, the interval $[\underline{\OG}({x},\widehat{x}),\overline{\OG}({x},\widehat{x})]$ is the smallest interval containing $g([{x},\widehat{x}])$. 
\end{enumerate}
\end{definition}
Given a nonlinear map $g:\real^n\to\real^m$, one can show that its minimal inclusion function $\OG^{\min} = \left[\begin{smallmatrix}\underline{\OG}^{\min}\\ \overline{\OG}^{\min}\end{smallmatrix}\right]$ is given by:
    \begin{align}\label{eq:tight}
    \underline{\OG}^{\min}_i({x},\widehat{x})= \inf_{z\in [{x},\widehat{x}]} g_i(z),\quad
    \overline{\OG}^{\min}_i({x},\widehat{x})=\sup_{z\in [{x},\widehat{x}]} g_i(z),
\end{align}  
for every $i\in \{1,\ldots,n\}$. The next theorem shows that the minimal  inclusion function provides sharper over-approximations compared to Lipschitz bounds. 

\begin{theorem}[Minimal inclusion functions]\label{thm:lip}
Consider the function $g:\real^n\to \real^m$ and the mapping $\OG^{\min} = \left[\begin{smallmatrix}\underline{\OG}^{\min}\\ \overline{\OG}^{\min}\end{smallmatrix}\right]$ defined by~\eqref{eq:tight}. Then, for every ${x}\le \widehat{x}$, we have
    \begin{align*}
        \|\underline{\OG}^{\min}({x},\widehat{x})- \overline{\OG}^{\min}({x},\widehat{x})\|_{\infty} \le \mathrm{Lip}_{\infty}(g) \|{x}-\widehat{x}\|_{\infty}.
    \end{align*}
\end{theorem}
\begin{proof} 
Let $i\in \{1,\ldots,k\}$ be such that $\|\overline{\OG}^{\min}(x,\widehat{x})-\underline{\OG}^{\min}(x,\widehat{x})\|_{\infty} = \left|\overline{\OG}^{\min}_i(x,\widehat{x})-\underline{\OG}^{\min}_i(x,\widehat{x})\right|$. 
Note that since $g$ is continuous and the box $[x,\widehat{x}]$ is compact, there exist $\eta^*,\xi^*\in [x,\widehat{x}]$ such that 
\begin{align*}
    \max_{y\in [x,\widehat{x}]}g_i(y) =g_i(\eta^*),\qquad \min_{y\in [x,\widehat{x}]}g_i(y)=g_i(\xi^*).
\end{align*}
This implies that $\|\overline{\OG}^{\min}(x,\widehat{x})-\underline{\OG}^{\min}(x,\widehat{x})\|_{\infty}= |g_i(\eta^*)-g_i(\xi^*)| \le  \|g(\xi^*)-g(\eta^*)\|_{\infty} \le 
   \mathrm{Lip}_{\infty}(g)\|\xi^*-\eta^*\|_{\infty} \le \mathrm{Lip}_{\infty}(g) \|x-\widehat{x}\|_{\infty}$. 
\end{proof}

\subsection{Designing Inclusion Functions}\label{subsec:compute}

\begin{figure}
    \centering
    \includegraphics[width=\columnwidth]{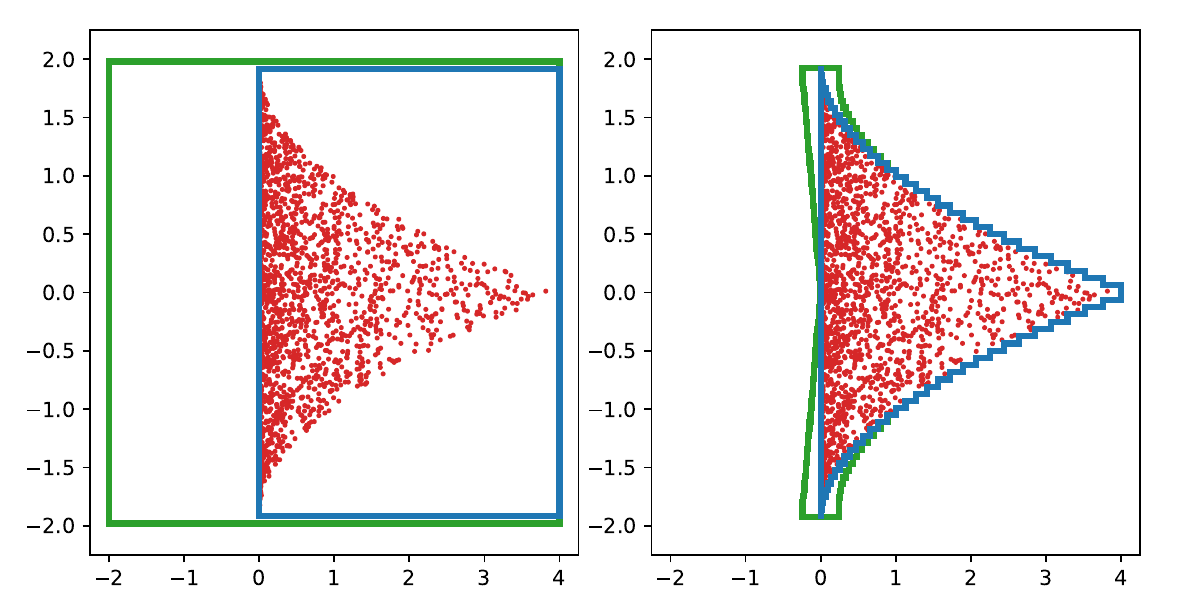}
    \caption{\textbf{Left:} \texttt{npinterval} is used to generate interval approximations for a function $g(x_1,x_2) = [(x_1 + x_2)^2, 4\sin((x_1 - x_2)/4)]^{\top} $ using two different natural inclusion functions. {\color{tab:blue} Blue:} using the inclusion functions for elementary functions $x\mapsto x^2$ and $x\mapsto \sin(x)$ in~\cite[Appendix A, Table 1]{AH-SJ-SC:23b}, 
    {\color{tab:green} Green:} rewriting $g(x_1,x_2) = [x_2^2 + 2x_1x_2 + x_2^2, 4\sin(x_1/4)\cos(x_2/4) - 4\cos(x_1/4)\sin(x_2/4)]^{\top}$ and obtaining a different natural inclusion function based on composition of elementary inclusion functions in~\cite[Appendix A, Table 1]{AH-SJ-SC:23b}. The approximations are generated using the initial set $[-1,1]\times[-1,1]$, and 2000 uniformly sampled ouptuts are shown in {\color{tab:red} red}. \textbf{Right:} The same function is analyzed, with the same two natural inclusion functions, but the initial set is partitioned into 1024 uniform sections, and the union of the interval approximations are shown.}
    \label{fig:interval_example1}
\end{figure}

Given a function $g$, finding a suitable inclusion function for $g$ is a central problem in interval analysis.  We review three approaches for constructing inclusion functions. 

    \paragraph*{Natural inclusion function}
    As we showed in the previous section, the minimal inclusion function can be computed using the optimization problem~\eqref{eq:tight}. However, for general functions, solving this optimization problem is not tractable. One feasible approach is to find the minimal inclusion function for a handful of elementary functions for which the optimization problem~\eqref{eq:tight} is solvable. Then, for an arbitrary function, a natural inclusion function can be obtained by expressing it as a finite composition of operators and elementary functions~\cite[Theorem 2]{LJ-MK-OD-EW:01}~\cite[Theorem 2]{AH-SJ-SC:23b}. Several software packages exist for automatically computing natural inclusions functions. In~\cite{AH-SJ-SC:23b}, we introduce the software package \texttt{npinterval}\footnote{The code for \texttt{npinterval} is available at \url{https://github.com/gtfactslab/npinterval}} that implements intervals as a native data type within the Python toolbox \texttt{numpy}~\cite{numpy}, yielding flexibility, efficiency through vectorization across arrays, and canonical constructions of natural inclusion functions. \change{An example of an inclusion function computed using \texttt{npinterval}~\cite{AH-SJ-SC:23b} is illustrated in Figure~\ref{fig:interval_example1}.}

    \paragraph*{Jacobian-based inclusion functions} Given a continuously differentiable function $g:\real^n\to \real^m$, one can use upper and lower bounds on the Jacobian of $g$ to construct inclusion functions for $g$. These inclusion functions are often derived from the Taylor expansion of the function $g$ around a certain {point}. Most commonly, this point is taken as the midpoint of an interval, leading to, in particular, the centered inclusion function~\cite[\S 2.4.3]{LJ-MK-OD-EW:01} and the mixed-centered inclusion function~\cite[\S 2.4.4]{LJ-MK-OD-EW:01}. In the next theorem, we develop a Jacobian-based inclusion function obtained by linearization around corners of an interval. As demonstrated in the sequel, such cornered inclusion functions turn out to be particularly amenable to analysis of closed-loop control systems. 

    \begin{proposition}[Jacobian-based cornered inclusion function]\label{thm:jacobian-basedIF}
    Given a continuously differentiable function $g:\real^n\to \real^m$ such that $Dg(z)\in [\underline{J}_{[x,\widehat{x}]},\overline{J}_{[x,\widehat{x}]}]$, for every $z\in [x,\widehat{x}]$. The function $\OG^{\mathrm{jac}}:\mathcal{T}^{2n}_{\ge 0} \to \mathcal{T}^{2m}_{\ge 0}$ defined by
    \begin{align*}
        \OG^{\mathrm{jac}}(x,\widehat{x}) = \begin{bmatrix}
           -[\underline{J}_{[x,\widehat{x}]}]^{-} & [\underline{J}_{[x,\widehat{x}]}]^- \\
           -[\overline{J}_{[x,\widehat{x}]}]^+ & [\overline{J}_{[x,\widehat{x}]}]^+
        \end{bmatrix}\begin{bmatrix}
            x\\ \widehat{x}
        \end{bmatrix} + \begin{bmatrix}
            g(x)\\
            g(x)
        \end{bmatrix}
    \end{align*}
is a inclusion function for $g$.
    \end{proposition}
    \begin{proof}
        Since $g$ is continuously differentiable, for every $z\in [x,\widehat{x}]$, we have  
            $g(z) = g(x) + \int_{0}^{1} Dg(tz + (1-t)x)(z-x) dt.$ 
        Thus, for every $x\in [x,\widehat{x}]$,
        \begin{align*}
            g(z) &  \le g(x)+ \overline{J}_{[x,\widehat{x}]}(z-x)\\ &= g(x)+[\overline{J}_{[x,\widehat{x}]}]^{+}(z-x) + [\overline{J}_{[x,\widehat{x}]}]^{-}(z-x)\\ & \le g(x)+[\overline{J}_{[x,\widehat{x}]}]^{+}(\widehat{x}-x) = \overline{\OG}^{\mathrm{jac}}(x,\widehat{x})
        \end{align*}
        where the first inequality holds because $z-x\ge \vect{0}_n$, the second equality holds by $\overline{J}_{[x,\widehat{x}]} = [\overline{J}_{[x,\widehat{x}]}]^{+}+ [\overline{J}_{[x,\widehat{x}]}]^{-}$, and the last inequality holds because $z-x\le \widehat{x}-x$ and $z-x\ge \vect{0}_n$. Similarly, one can show that, for every $z\in [x,\widehat{x}]$, 
            $g(z) \ge g(x) + [\underline{J}_{[x,\widehat{x}]}]^-(\widehat{x}-x) = \underline{\OG}^{\mathrm{jac}}(x,\widehat{x}), 
        $ %
        completing the proof.    
    \end{proof}

\begin{remark}(The role of corner points)
   The proof of Proposition~\ref{thm:jacobian-basedIF} is based on Taylor expansion of $g$ around the corner point $x$. The proposition extends straightforwardly to the Taylor expansion of 
$g$ around other corner points, i.e., around any point $y$ such that $y_i\in \{x_i,\widehat{x}_i\}$, resulting in different inclusion functions.  
   Both the Jacobian-based cornered inclusion function (Proposition~\ref{thm:jacobian-basedIF}) and the Jacobian-based centered inclusion function~\cite[\S 2.4]{LJ-MK-OD-EW:01} are obtained using the Taylor expansion of the original function. In general, they %
    are not comparable (cf. Example~\ref{exm:corner} and Table~\ref{tab:1}). However, for functions that are monotone in their entries, the Jacobian-based cornered inclusion function will lead to the minimal inclusion function while the Jacobian-based centered inclusion function is often non-minimal. %
    Figure~\ref{fig:compareinclusionfunctions} provides a pictorial comparison between the minimal, the centered, and the cornered inclusion functions.

\begin{figure}
    \centering
    \includegraphics[width=\columnwidth]{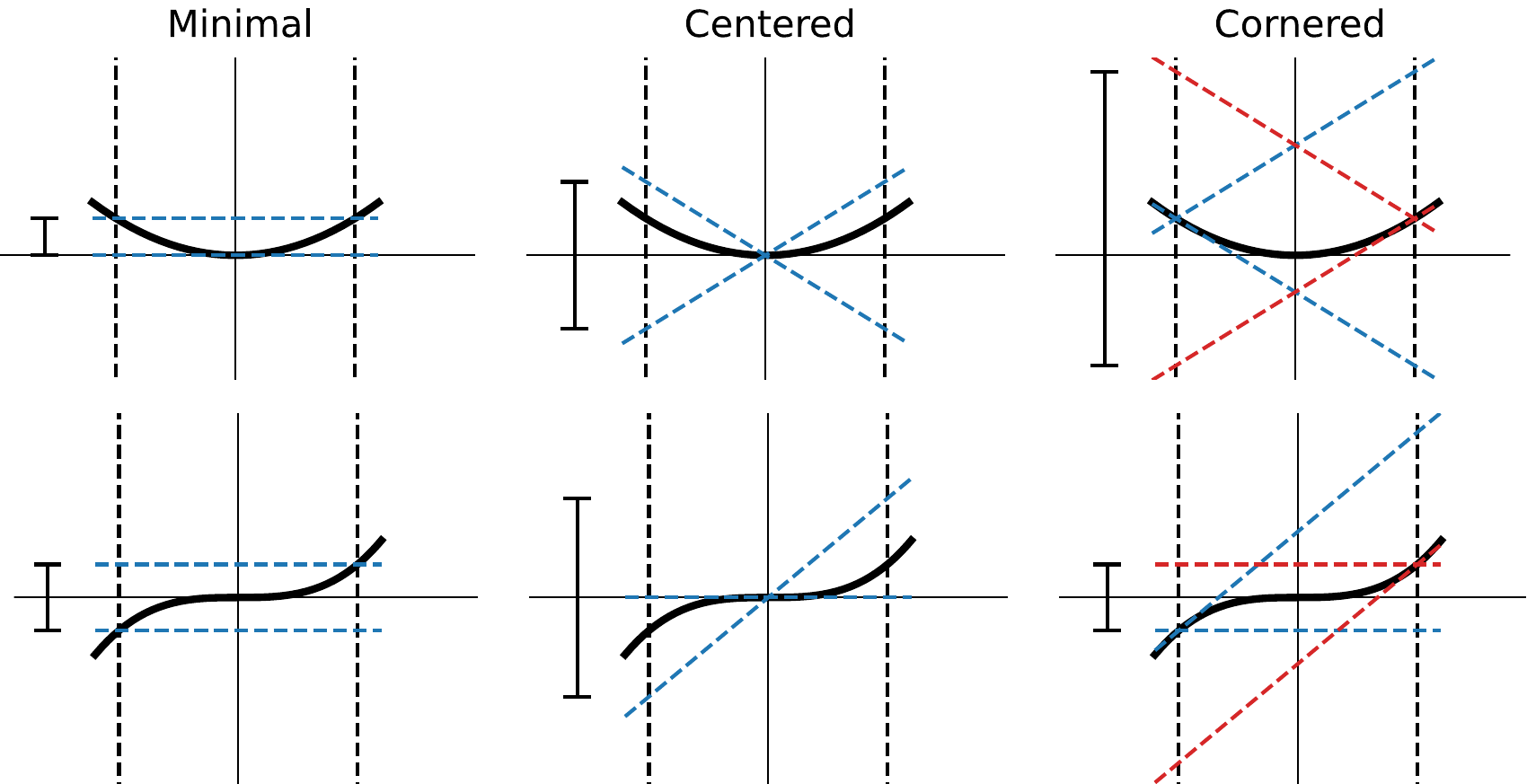}
    \vspace{-0.5cm}
    \caption{Pictorial comparison between the minimal (equation~\eqref{eq:tight}), the Jacobian-based centered (\cite[\S 2.4.3]{LJ-MK-OD-EW:01}), and the Jacobian-based cornered (Proposition~\ref{thm:jacobian-basedIF}) inclusion functions for $f(x)=x^2$ and $f(x)=x^3$ on the interval $[-1,1]$. Note that, for the monotone function $f(x)=x^3$, the intersection of the two Jacobian-based cornered inclusion functions (shown by {\color{tab:blue}blue} and {\color{tab:red}red} dashed lines) leads to the minimal inclusion function.}
    \label{fig:compareinclusionfunctions}
\end{figure}

\end{remark}

The accuracy of the Jacobian-based cornered inclusion function can be improved at a cost of a slightly more complicated formulation using mixed differentiation in the same way that the Jacobian-based centered inclusion function is improved with the mixed-centered inclusion function~\cite[\S 2.4.4]{LJ-MK-OD-EW:01}. The idea is to use a step-by-step differentiation of the function with respect to each variable. 

\begin{proposition}[Mixed Jacobian-based cornered inclusion function]\label{thm:mixed}

Given a continuously differentiable function $g:\real^n\to \real^m$ such that $Dg(z)\in [\underline{J}_{[x,\widehat{x}]},\overline{J}_{[x,\widehat{x}]}]$, for every $z\in [x,\widehat{x}]$. The function $\OG^{\mathrm{jac-m}}:\mathcal{T}^{2n}_{\ge 0} \to \mathcal{T}^{2m}_{\ge 0}$ defined by
    \begin{align*}
        \OG^{\mathrm{jac-m}}(x,\widehat{x}) = \begin{bmatrix}
           -[\ulM]^{-} & [\ulM]^- \\
           -[\olM]^+ & [\olM]^+
        \end{bmatrix}\begin{bmatrix}
            x\\ \widehat{x}
        \end{bmatrix} + \begin{bmatrix}
            g(x)\\
            g(x)
        \end{bmatrix}
    \end{align*}
is an inclusion function for $g$, where the $i$th columns of $\ulM$ and $\olM$ are defined by
\begin{align*}
    \ulM_i = (\underline{J}_{[x,\widehat{x}_{[R_i:x]}]})_i,\qquad
    \olM_i = 
    (\overline{J}_{[x,\widehat{x}_{[R_i:x]}]})_i, 
\end{align*}
with $R_i=\{i+1,\ldots,n\}$, for every $i\in \{1,\ldots,n\}$. Moreover, for every $x\le z\le \widehat{x}$, we have
\begin{align*}
    \underline{\OG}^{\mathrm{jac}}(x,\widehat{x})\le \underline{\OG}^{\mathrm{jac-m}}(x,\widehat{x}),\qquad \overline{\OG}^{\mathrm{jac-m}}(x,\widehat{x})\le \overline{\OG}^{\mathrm{jac}}(x,\widehat{x}).
\end{align*}
\end{proposition}
\begin{proof}  
Since $g$ is continuously differentiable, for every $z\in [x,\widehat{x}]$, we have 
            $g(z_{[R_1:x]}) = g(x) + \int_{0}^{1} D_{x_1}g(tz_{[R_1:x]} + (1-t)x)(z_1-x_1) dt$.  
        Thus, for every $x\in [x,\widehat{x}]$, 
            $g(z_{[R_1:x]})   \le g(x)+ \overline{J}_{[x,\widehat{x}_{[R_1:x]}]}(z_1-x_1) \le g(x)+[\overline{J}_{[x,\widehat{x}_{[R_1:x]}]}]^{+}(\widehat{x}_1-x_1). $ 
        where the first inequality holds because $z_1-x_1\ge \vect{0}_n$ and the second inequality holds because $z_1-x_1\le \widehat{x}_1-x_1$ and $z_1-x_1\ge \vect{0}_n$. \change{By successively applying the above procedure on $g(z_{[R_1:x]})$, one can show that $g(z)\leq \overline{\OG}^{\mathrm{jac-m}}(x,\widehat{x})$. Similarly, one can show that $g(z)\geq \underline{\OG}^{\mathrm{jac-m}}(x,\widehat{x})$ and thus $\OG^{\mathrm{jac-m}}$ is an inclusion function for $g$.} Finally, one can show that $ \underline{J}_{[x,\widehat{x}]}\le \underline{M}$ and $\overline{M}\le \overline{J}_{[x,\widehat{x}]}$. This implies that $-[\underline{M}]^{-}x + [\underline{M}]^{-}\widehat{x} \le -[\underline{J}_{[x,\widehat{x}}]^{-} x + [\overline{J}_{[x,\widehat{x}}]^{-}\widehat{x}$. As a result $\underline{\OG}^{\mathrm{jac}}(x,\widehat{x}) \le \underline{\OG}^{\mathrm{jac-m}}(x,\widehat{x})$. Similarly, one can show that $\overline{\OG}^{\mathrm{jac-m}}(x,\widehat{x}) \le \overline{\OG}^{\mathrm{jac}}(x,\widehat{x})$.
\end{proof}

\begin{example}[Cornered inclusion functions]\label{exm:corner}
Consider the function $f:\real^2\to\real^2$ defined by $f(x_1,x_2)=\begin{bmatrix}
(x_1+x_2)^2\\ x_1+x_2+2x_1x_2
    \end{bmatrix}$. We compute the natural inclusion function, the (mixed) Jacobian-based  cornered inclusion function, and the (mixed) Jacobian-based centered inclusion function for $f$ on the interval $[-0.1,0.1]\times [-0.1,0.1]$. The results are shown in Table~\eqref{tab:1}, for runtimes averaged over $10\,000$ runs. We observe that the output interval estimated by the mixed Jacobian-based centered (cornered) is tighter than the interval obtained from the Jacobian-based centered (cornered) inclusion function. This observation is consistent with the comparison between mixed and non-mixed Jacobian-based cornered inclusion functions in Proposition~\ref{thm:mixed}. However, there is generally no ranking between other methods, as each can perform better in different scenarios.

    \begin{table}[h!]
    \centering
    \begin{tabular}{||c | c | c||}
        \hline
        Method & Output Interval & Runtime (s)\\
        \hline\hline
Natural & $[0, 0.04]\times [-0.22, 0.22]$ & $1.93\times10^{-6}$ \\
       
Centered& $[-0.08, 0.08]\times [-0.24, 0.24]$ & $1.64\times10^{-5}$ \\
      
Mixed Centered & $[-0.06, 0.06]\times [-0.22, 0.22]$ &  $8.38\times10^{-5}$  \\
       
Cornered &$[-0.12, 0.16]\times [-0.18, 0.22]$ & $6.24\times10^{-5}$ \\
        
Mixed Cornered & $[-0.08, 0.08]\times[-0.18, 0.22]$ & $2.28\times10^{-4}$ \\
        \hline
    \end{tabular}
    \caption{Comparison of inclusion function methods.} 
    \label{tab:1}
\end{table}

\end{example}

    \paragraph*{Decomposition-based inclusion functions}

Another approach for constructing inclusion functions is by decomposing the function into the sum of increasing and decreasing parts~\cite{JLG-KPH:94,HLS:08}. This separation is often realized using a decomposition function.  

\begin{definition}[Decomposition function]\label{def:decomp}
Given a function $g:\real^n\to \real^m$,  a \emph{decomposition  function} for $g$ is a map $d:\mathcal{T}^{2n}\to \real^{m}$ that satisfies
\begin{enumerate}[label=(\alph*)]
    \item\label{p1} $g(x) = d(x,x)$, for every $x\in \real^n$; 
    \item\label{p2} $d(x,y)\le d(\widehat{x},y)$, for every $y\in \real^n$ and every $x\le \widehat{x}$;  
    \item\label{p3} $d(x,\widehat{y})\le d(x,y)$, for every $x\in \real^n$ and every $y\le \widehat{y}$.
\end{enumerate}
\end{definition}

One can show that every continuous function has at least one decomposition function. Given a map $g:\real^n\to \real^m$, we define the tight decomposition function for $g$ as the mapping $d^{\mathrm{tight}}:\mathcal{T}^{2n}\to \real^{2m}$ given by 
\begin{align}\label{eq:tight-decomp}
    d_i^{\mathrm{tight}}(x,\widehat{x}) = \begin{cases}
        \min_{z\in [x,\widehat{x}]} g_i(z) & x\le \widehat{x}\\
         \max_{z\in [\widehat{x},x]} g_i(z) & \widehat{x}\le x
    \end{cases}
\end{align}
for every $i\in \{1,\ldots,n\}$. In general, the decomposition function for function $g$ is not unique. The next theorem, whose proof is straightforward, shows how to construct inclusion functions from decomposition functions.

   \begin{proposition}[Decomposition-based inclusion functions]\label{thm:decomp}
Consider a map $g:\real^n\to \real^m$ with a function $d:\mathcal{T}^{2n}\to \real^m$. 
\change{ 
Then, the map $\sfG^d:\calT^{2n}\to\R^{2m}$ defined by 
\[
\OG^d(x,\widehat{x}) = \left[\begin{smallmatrix}d(x,\widehat{x})\\ d(\widehat{x},x)\end{smallmatrix}\right],\quad \mbox{ for all } x\le \widehat{x},
\]
is a thin monotone inclusion function for $g$ if and only if $d$ is a decomposition function for $g$. Moreover, $\sfG^d$ is the minimal inclusion function~\eqref{eq:tight} if and only if $d$ is the tight decomposition function~\eqref{eq:tight-decomp}.
}
\end{proposition}

\begin{remark}[Comparison with the literature]  For a vector field $f:\real^n\to \real^n$, the decomposition function has already been studied in  mixed monotone theory (cf. \cite{DG-VL:87},~\cite[Definition 2]{SC-MA:15b},~\cite[Definition 4]{LY-OM-NO:19}, and~\cite{HLS:06}). Definition~\ref{def:decomp} extends this classical notion to arbitrary functions. 
\end{remark}

One can combine two inclusion functions to obtain another inclusion function with tighter estimates for the output of the original function. 

\begin{proposition}[Intersection of inclusion functions]
\label{prop:intersect}
Given a map $g:\real^n\to \real^m$ with two inclusion functions $\OG_1 = \left[\begin{smallmatrix}\underline{\OG}_1\\ \overline{\OG}_1\end{smallmatrix}\right]:\mathcal{T}^{2n}_{\ge 0}\to \mathcal{T}^{2m}_{\ge 0}$ and $\OG_2=\left[\begin{smallmatrix}\underline{\OG}_2\\ \overline{\OG}_2\end{smallmatrix}\right]:\mathcal{T}^{2n}_{\ge 0}\to \mathcal{T}^{2m}_{\ge 0}$. Define the mapping $\OG_1\wedge \OG_2 = \left[\begin{smallmatrix}\underline{\OG}_1\wedge \underline{\OG}_2\\ \overline{\OG}_1\wedge \overline{\OG}_2\end{smallmatrix}\right]:\mathcal{T}^{2n}_{\ge 0}\to \mathcal{T}^{2m}_{\ge 0}$  by 
    \begin{align*}
      [\underline{\OG}_1\wedge \underline{\OG}_2](x,\widehat{x}) &= \max\{ \underline{\OG}_1(x,\widehat{x}), \underline{\OG}_2(x,\widehat{x})\},\\
      [\overline{\OG}_1\wedge \overline{\OG}_2](x,\widehat{x}) &= \min\{ \overline{\OG}_1(x,\widehat{x}), \overline{\OG}_2(x,\widehat{x})\}.
    \end{align*}
    Then $\OG_1\wedge \OG_2$ is an inclusion function for $g$ and it provides over-approximations of the output of $g$ which are tighter than both $\OG_1$ and $\OG_2$. 
\end{proposition}
\begin{proof}
  Note that both $\OG_1$ and $\OG_2$ are inclusion functions for $g$. Thus, for every $x\le z\le \widehat{x}$, we have $\OG_1(x,\widehat{x}) \le g(z)$ and $\OG_2(x,\widehat{x}) \le g(z)$. This implies that $\max\{\underline{\OG}_1(x,\widehat{x}),\underline{\OG}_2(x,\widehat{x})\} = [\underline{\OG}_1\wedge \underline{\OG}_2](x,\widehat{x}) \le g(z)$. 
  Similarly, one can show that $g(z)\le [\overline{\OG}_1\wedge\overline{\OG}_2](x,\widehat{x})$.
\end{proof}

\subsection{Convergence of inclusion functions}

We now define the convergence rate of an inclusion function which can be used to capture the accuracy of the interval estimations. 

\begin{definition}[Convergence rate of inclusion functions]\label{def:conv-rate-inclusion}
   Consider the mapping $g:\real^n\to \real^m$ with an inclusion function $\OG=\left[\begin{smallmatrix}\underline{\OG}\\ \overline{\OG}\end{smallmatrix}\right]:\mathcal{T}_{\ge 0}^{2n}\to \mathcal{T}_{\ge 0}^{2m}$. The rate of convergence of $\OG$ is the largest $\alpha\in \real_{\ge 0}\cup\{\infty\}$ such that 
    \begin{multline}\label{eq:conv-rate}
        \|\overline{\OG}(x,\widehat{x})-\underline{\OG}(x,\widehat{x})\|_{\infty} - \max_{z,y\in [x,\widehat{x}]}\|g(z)-g(y)\|_{\infty} \\\le \beta \|\widehat{x}-x\|_{\infty}^{\alpha}
    \end{multline}
    for some $\beta\in \real_{\ge 0}$ and for every $x\le \widehat{x}$. 
\end{definition}

Roughly speaking, the rate of convergence of an inclusion function $\OG$ shows how accurately it estimates the range of the original function $g$. Indeed, for the minimal inclusion function, the interval $[\underline{\OG}(x,\widehat{x}), \overline{\OG}(x,\widehat{x})]$ is the tightest interval containing $g([x,\widehat{x}])$. Thus the left hand side of inequality~\eqref{eq:conv-rate} is always zero and the rate of convergence of the minimal inclusion function is $\alpha =\infty$. It is shown that, in general, the convergence rate of the natural inclusion functions is $\alpha= 1$~\cite[Lemma 4.1]{REM-RBK-MJC:09} and the convergence rate of the Jacobian-based inclusion functions are $\alpha=2$~\cite[\S 2.4.5]{LJ-MK-OD-EW:01}.

\subsection{Inclusion functions for Neural Networks}\label{sec:nninclusion}

Given a neural network $N$ of the form \eqref{eq:NN} and any interval $[y,\widehat{y}] \subseteq \real^n$, as described in the sequel, our framework proposes using existing neural network verification algorithms to construct a $[y,\widehat{y}]$-localized inclusion function of the form $\ON_{[y,\widehat{y}]}=\left[\begin{smallmatrix}\underline{\ON}_{[y,\widehat{y}]}\\ \overline{\ON}_{[y,\widehat{y}]}\end{smallmatrix}\right]: \mathcal{T}^{2n}_{\geq 0}\to \mathcal{T}^{2p}_{\geq 0}$ for $N$ satisfying
    \begin{align}\label{eq:NNverifier}
        \underline{\ON}_{[y,\widehat{y}]}(x,\widehat{x}) \leq N(x) \leq \overline{\ON}_{[y,\widehat{y}]}(x,\widehat{x}),
    \end{align}
    for any $x\in[x,\widehat{x}]\subseteq[y,\widehat{y}]$.   
    A large number of the existing neural network verification algorithms can provide bounds of the form~\eqref{eq:NNverifier} for the output of the neural networks, including CROWN~\cite{HZ-etal:18}, LipSDP \cite{MF-MM-GJP:22}, and IBP \cite{SG-etal:19}. In particular, some neural network verification algorithms can provide \emph{affine} $[y,\widehat{y}]$-localized inclusion functions for $N$. Examples of these neural network verification algorithms include CROWN and its variants~\cite{HZ-etal:18}. Given an interval $[y,\widehat{y}]$, these algorithms provide a tuple $(\underline{C},\overline{C},\underline{d},\overline{d})$ defining affine upper and lower bounds for the output of the neural network
\begin{gather} \label{eq:crown}
    \underline{C}(y,\widehat{y})x + \underline{d}(y,\widehat{y}) \leq N(x) \leq \overline{C}(y,\widehat{y})x + \overline{d}(y,\widehat{y}),
\end{gather}
for every $x\in [y,\widehat{y}]$. Using these linear bounds, we can construct the affine $[y,\widehat{y}]$-localized inclusion function for $N$:
\begin{gather} \label{eq:crownif}
    \begin{aligned}
    \underline{\ON}_{[y,\widehat{y}]}(x,\widehat{x}) &= [\underline{C}(y,\widehat{y})]^+x + [\underline{C}(y,\widehat{y})]^-\widehat{x} + \underline{d}(y,\widehat{y}),\\
        \overline{\ON}_{[y,\widehat{y}]}(x,\widehat{x}) &= [\overline{C}(y,\widehat{y})]^+\widehat{x} + [\overline{C}(y,\widehat{y})]^-x + \overline{d}(y,\widehat{y})
    \end{aligned}
\end{gather}
for any $[x,\widehat{x}]\subseteq[y,\widehat{y}]$.  

\begin{remark}
\;
\begin{enumerate}
    \item \textit{(Computational complexity):} the computational complexity of finding the $[y,\widehat{y}]$-localized inclusion function $\ON_{[y,\widehat{y}]}=\left[\begin{smallmatrix}\underline{\ON}_{[y,\widehat{y}]}\\ \overline{\ON}_{[y,\widehat{y}]}\end{smallmatrix}\right]$ depends on the neural network verification algorithm. For instance, for a neural network with $k$-layer and $M$ neurons per layer, the computational complexity of CROWN~\cite{HZ-etal:18} in finding the bounds of the form~\eqref{eq:NNverifier} is $\mathcal{O}(k^2M^{3})$.
    \item \textit{(Inclusion function):} Since the $[y,\widehat{y}]$-localized inclusion function $\ON_{[y,\widehat{y}]}=\left[\begin{smallmatrix}\underline{\ON}_{[y,\widehat{y}]}\\ \overline{\ON}_{[y,\widehat{y}]}\end{smallmatrix}\right]$ can be obtained for every $[y,\widehat{y}]$, we can define an inclusion function $\ON=\left[\begin{smallmatrix}\underline{\ON}\\ \overline{\ON}\end{smallmatrix}\right]$ for the neural network $N$ by $\ON(x,\widehat{x})= \ON_{[x,\widehat{x}]}(x,\widehat{x})$, for all $x\le \widehat{x}$.
\end{enumerate}
\end{remark}

\section{Interval Reachability of Dynamical Systems}\label{sec:IA-d}
In this section, we study reachability of the dynamical system~\eqref{eq:dynamics-c} using interval analysis. Our key idea is to embed the original dynamical system~\eqref{eq:dynamics-c} into an \emph{embedding system} on $\real^{2n}$ and use trajectories of the embedding system to study the propagation of the state bounds with time. Our main goal is to use the notion of inclusion function (cf. Definition~\ref{def:embedding}) to design embedding systems for dynamical systems. Consider the dynamical system~\eqref{eq:dynamics-c} with an inclusion function $\OF = \left[\begin{smallmatrix}\underline{\OF}\\ \overline{\OF}\end{smallmatrix}\right]:\mathcal{T}^{2n}_{\ge 0}\times \mathcal{T}^{2q}_{\ge 0}\to \mathcal{T}^{2n}_{\ge 0}$ for $f$. We define the associated \emph{embedding system} on $\real^{2n}$:
\begin{align}\label{eq:ct-embedingsystem}
    \dot{x}_i &= \underline{\OF}_i(x,\widehat{x}_{[i:x]},w,\widehat{w}), \nonumber\\ \dot{\widehat{x}}_i
 &= \overline{\OF}_i(x_{[i:\widehat{x}]},\widehat{x},w,\widehat{w}),
\end{align}
for every $i\in \{1,\ldots,n\}$. 
For instance, using the minimal inclusion function $\OF^{\min} = \left[\begin{smallmatrix}\underline{\OF}^{\min}\\ \overline{\OF}^{\min}\end{smallmatrix}\right]$ for $f$, the associated embedding system is given by:
\begin{align}\label{eq:tightembedding-c}
\dot{x}_i = & \min_{\substack{z\in [x,\widehat{x}],\; z_i = x_i,\\ u\in [w,\widehat{w}]}} f_i(z,u),\nonumber\\
\dot{\widehat{x}}_i =  &  \max_{\substack{z\in [x,\widehat{x}],\; z_i = \widehat{x}_i,\\ u\in [w,\widehat{w}]}} f_i(z,u),
\end{align}
for every $i\in \{1,\ldots,n\}$. Given the disturbance set $\mathcal{W}=[\underline{w},\overline{w}]$ and the initial set $\mathcal{X}_0=[\underline{x}_0,\overline{x}_0]$, one can use a single trajectory of the embedding system~\eqref{eq:ct-embedingsystem} to obtain over-approximations of reachable sets of the system~\eqref{eq:dynamics-c}. 
\change{Moreover, the embedding system obtained from the minimal inclusion function (i.e., the  dynamical system~\eqref{eq:tightembedding-c}) provides the best reachable set over-approximations at any time $t\geq 0$, compared to any other embedding system constructed using an inclusion function for vector field $f$ in the dynamical system~\eqref{eq:dynamics-c}.}

\begin{proposition}[Reachability using embedding systems]\label{thm:reach-c}
Consider the dynamical system~\eqref{eq:dynamics-c} with the disturbance set $[\underline{w},\overline{w}]$ and the initial set $\mathcal{X}_0=[\underline{x}_0,\overline{x}_0]$. Suppose that  $\OF=\left[\begin{smallmatrix}\underline{\OF}\\ \overline{\OF}\end{smallmatrix}\right]$ is an inclusion function for $f$, and $t\mapsto \left[\begin{smallmatrix}\underline{x}(t)\\ \overline{x}(t)\end{smallmatrix}\right]$ and $t\mapsto \left[\begin{smallmatrix}\underline{x}^{\min}(t)\\ \overline{x}^{\min}(t)\end{smallmatrix}\right]$ are the trajectories of embedding systems~\eqref{eq:ct-embedingsystem} and~\eqref{eq:tightembedding-c}, starting from $\left[\begin{smallmatrix}\underline{x}_0\\ \overline{x}_0\end{smallmatrix}\right]$ with fixed disturbance $\left[\begin{smallmatrix}w\\ \widehat{w}\end{smallmatrix}\right] = \left[\begin{smallmatrix}\underline{w}\\ \overline{w}\end{smallmatrix}\right]$. Then, for every $t\in\real_{\ge 0}$, we have 
    \begin{align*}
    \mathcal{R}_f(t,[\underline{x}_0,\overline{x}_0], [\underline{w},\overline{w}])\subseteq [\underline{x}^{\min}(t),\overline{x}^{\min}(t)]\subseteq [\underline{x}(t),\overline{x}(t)]. 
\end{align*}
\end{proposition}
\begin{proof}
We first show that the dynamical system~\eqref{eq:tightembedding-c} is a monotone dynamical system with respect to the southeast partial order $\le_{\mathrm{SE}}$. Let $\left[\begin{smallmatrix}x\\ \widehat{x}\end{smallmatrix}\right] \le_{\mathrm{SE}} \left[\begin{smallmatrix}y\\ \widehat{y}\end{smallmatrix}\right]$ be such that $x_i=y_i$ and let $\left[\begin{smallmatrix}w\\ \widehat{w}\end{smallmatrix}\right] \le_{\mathrm{SE}} \left[\begin{smallmatrix}v\\ \widehat{v}\end{smallmatrix}\right]$. Then, we have $[y,\widehat{y}]\subseteq [x,\widehat{x}]$ and $[v,\widehat{v}]\subseteq [w,\widehat{w}]$, and we can compute
\begin{align*}
   \underline{\OF}^{\min}_i(x,\widehat{x},w,\widehat{w}) &= \min_{\substack{z\in [x,\widehat{x}],\; z_i = x_i,\\ u\in [w,\widehat{w}]}} f_i(z,u) \\ & \le \min_{\substack{z\in [y,\widehat{y}],\; z_i = y_i,\\ u\in [v,\widehat{v}]}} f_i(z,u) =  \underline{\OF}^{\min}_i(y,\widehat{y},v,\widehat{v}),
\end{align*}
for every $i\in \{1,\ldots,n\}$. Similarly, let $\left[\begin{smallmatrix}x\\ \widehat{x}\end{smallmatrix}\right]\le_{\mathrm{SE}} \left[\begin{smallmatrix}y\\ \widehat{y}\end{smallmatrix}\right]$ be such that $\widehat{x}_i=\widehat{y}_i$ and let $\left[\begin{smallmatrix}w\\ \widehat{w}\end{smallmatrix}\right] \le_{\mathrm{SE}} \left[\begin{smallmatrix}v\\ \widehat{v}\end{smallmatrix}\right]$. Then, we have $[y,\widehat{y}]\subseteq [x,\widehat{x}]$ and $[v,\widehat{v}]\subseteq [w,\widehat{w}]$, and {\color{black}similar reasoning implies $\overline{\OF}_i^{\min}(y,\widehat{y},v,\widehat{v})\leq \overline{\OF}^{\min}_{i}(x,\widehat{x},w,\widehat{w})$,}
for every $i\in \{1,\ldots,n\}$. As a result, the embedding system~\eqref{eq:tightembedding-c} is monotone with respect to $\le_{\mathrm{SE}}$. Let $x_0\in [\underline{x}_0,\overline{x}_0]$ and $u\in [\underline{w},\overline{w}]$ and $t\mapsto x_u(t)$ be the trajectory of the system~\eqref{eq:dynamics-c} starting from $x_0$ with disturbance $u$. Note that $t\mapsto \left[\begin{smallmatrix}{x}_u(t)\\ {x}_u(t)\end{smallmatrix}\right]$ is a solution of the embedding system~\eqref{eq:tightembedding-c} with the initial condition $\left[\begin{smallmatrix}{x}_0\\ {x}_0\end{smallmatrix}\right]$ and the disturbance $\left[\begin{smallmatrix}u\\ u\end{smallmatrix}\right]$. Moreover, we know that $\left[\begin{smallmatrix}\underline{x}_0\\ \overline{x}_0\end{smallmatrix}\right] \le_{\mathrm{SE}} \left[\begin{smallmatrix}x_0\\ x_0\end{smallmatrix}\right]$ and $\left[\begin{smallmatrix}\underline{w}\\ \overline{w}\end{smallmatrix}\right] \le_{\mathrm{SE}} \left[\begin{smallmatrix}u\\ u\end{smallmatrix}\right]$. Therefore, by monotonicity of the embedding system~\eqref{eq:tightembedding-c} with respect to $\le_{\mathrm{SE}}$, we get $\left[\begin{smallmatrix}\underline{x}^{\min}(t)\\ \overline{x}^{\min}(t)\end{smallmatrix}\right] \le_{\mathrm{SE}}\left[\begin{smallmatrix}{x}_u(t)\\ {x}_u(t)\end{smallmatrix}\right]$, for every $t\in \real_{\ge 0}$.  
This implies that $\mathcal{R}_f(t,[\underline{x}_0,\overline{x}_0],[\underline{w},\overline{w}]) \subseteq [\underline{x}^{\min}(t),\overline{x}^{\min}(t)]$, for every $t\in \real_{\ge 0}$. On the other hand, $\OF^{\min}$ is the minimal inclusion function for $f$ and therefore, for every $x\le \widehat{x}$ and every $w\le \widehat{w}$, we have $\left[\begin{smallmatrix}\underline{\OF}(x,\widehat{x},w,\widehat{w})\\ \overline{\OF}(x,\widehat{x},w,\widehat{w})\end{smallmatrix}\right] \le_{\mathrm{SE}} \left[\begin{smallmatrix}\underline{\OF}^{\min}(x,\widehat{x},w,\widehat{w})\\ \overline{\OF}^{\min}(x,\widehat{x},w,\widehat{w})\end{smallmatrix}\right]$.
Thus, for every $i\in \{1,\ldots,n\}$, every $x\le \widehat{x}$, and every $w\le \widehat{w}$,
\begin{align*}
    \underline{\OF}_i(x,\widehat{x}_{[i:x]},w,\widehat{w}) &\le \min_{\substack{z\in [x,\widehat{x}],\; z_i = x_i,\\ \xi\in [w,\widehat{w}]}} f_i(z,\xi), \\
    \overline{\OF}_i(x_{[i:\widehat{x}]},\widehat{x},w,\widehat{w}) &\ge \max_{\substack{z\in [x,\widehat{x}],\; z_i = \widehat{x}_i,\\ \xi\in [w,\widehat{w}]}} f_i(z,\xi).
\end{align*}
Now, we can use the classical monotone comparison Lemma~\cite[Theorem 3.8.1]{ANM-LH-DL:08-new} to obtain $\left[\begin{smallmatrix}\underline{x}(t)\\ \overline{x}(t)\end{smallmatrix}\right]
  \le_{\mathrm{SE}}    \left[\begin{smallmatrix}\underline{x}^{\min}(t)\\ \overline{x}^{\min}(t)\end{smallmatrix}\right]$, for every $t\in \real_{\ge 0}$.
This implies that $[\underline{x}^{\min}(t),\overline{x}^{\min}(t)] \subseteq [\underline{x}(t),\overline{x}(t)]$, for every $t\in \real_{\ge 0}$, and completes the proof. 
\end{proof}

\begin{remark}\label{rem:localized}
\;
\begin{enumerate}
    \item\label{p1:localized} It is straightforward to extend Proposition~\ref{thm:reach-c} to the case when the mapping $\OF$ is a $[y,\widehat{y}]\times \real^q$-localized inclusion function for $f$. In this setting, the results of Proposition~\ref{thm:reach-c} hold as long as we have $[\underline{x}(t),\overline{x}(t)]\subseteq [y,\widehat{y}]$. 
    \item Proposition~\ref{thm:reach-c} can be considered as a generalization of~\cite[Proposition 6]{PJM-AD-MA:19} and~\cite[Proposition 1]{MA-MD-SC:21}. Indeed, in the special case that $\OF$ is a decomposition-based inclusion function constructed from the decomposition function $d$, one can recover the embedding system
\begin{align*}
\dot{x}_i &= d(x,\widehat{x}_{[i:x]},w,\widehat{w}),\\  \dot{\widehat{x}}_i & = 
 d(\widehat{x},x_{[i:\widehat{x}]},\widehat{w},w)
\end{align*}
for every $i\in \{1,\ldots,n\}$. This embedding system is identical to~\cite[Equation (7)]{PJM-AD-MA:19} and~\cite[Equation (3)]{MA-MD-SC:21}. We highlight that this extension is crucial for our framework as the inclusion functions we obtain for learning-based components are neither thin nor decomposition-based. 
    
    \item Proposition~\ref{thm:reach-c} can alternatively be proved using~\cite[Theorem 2]{JKS-PIB:13}. However, our proof of Proposition~\ref{thm:reach-c} is different in that it uses monotone system theory~\cite{DA-EDS:03,EDS:01} and classical monotone comparison Lemma~\cite[Theorem 3.8.1]{ANM-LH-DL:08-new}. Indeed, compared to~\cite[Theorem 2]{JKS-PIB:13}, our proof techniques allow to compare accuracy of different over-approximations of reachable sets.    

\item  \change{For the dynamical system~\eqref{eq:dynamics-c} with locally Lipschitz vector field $f$, the minimal inclusion function for $f$ as defined in equation~\eqref{eq:tight} exists. Therefore, the reachability framework from Proposition~\ref{thm:reach-c} is applicable to locally Lipschitz nonlinear systems, regardless of monotonicity of the original dynamics. }

    \end{enumerate}
    \end{remark}

We now study the accuracy of the over-approximations provided in Proposition~\ref{thm:reach-c} using the notion of convergence rate of embedding systems.

\begin{definition}[Convergence rate of embedding systems]\label{def:conv-rate-emb}
Consider the dynamical system~\eqref{eq:dynamics-c}
with the embedding system~\eqref{eq:ct-embedingsystem}. Then the convergence rate of the embedding system~\eqref{eq:ct-embedingsystem} is the largest value of $\alpha\in \real_{\ge 0}\cup \{\infty\}$ such that 
\begin{multline*}
    \|\overline{x}(t)-\underline{x}(t)\|_{\infty} - \max_{y,z\in [\underline{x}_0,\overline{x}_0]}\|\phi_f(t,y)-\phi_f(t,z)\|_{\infty} \\ \le M(t) \|\overline{x}_0-\underline{x}_0\|^{\alpha}_{\infty}
\end{multline*}
for some $M:\mathbb{R}_{\ge 0}\to \real_{\ge 0}$ and for $t\mapsto \left[\begin{smallmatrix}\underline{x}(t)\\ \overline{x}(t)\end{smallmatrix}\right]$ being the trajectory of the embedding system~\eqref{eq:ct-embedingsystem} starting from any $\left[\begin{smallmatrix}\underline{x}_0\\ \overline{x}_0\end{smallmatrix}\right]$ with fixed disturbance $\left[\begin{smallmatrix}w\\ \widehat{w}\end{smallmatrix}\right] = \vect{0}_{2n}$.    
\end{definition}

Definition~\ref{def:conv-rate-emb} raises an important question: what is the connection between the convergence rate of the embedding system~\eqref{eq:ct-embedingsystem} and the convergence rate of its associated inclusion function $\OF$ (as in Definition~\ref{def:conv-rate-inclusion})? In the special case when $f$ is a discrete-time monotone vector field, one can choose the inclusion function $\OF(\underline{x},\overline{x}) = \left[\begin{smallmatrix}f(\underline{x})\\ f(\overline{x})\end{smallmatrix}\right]$ for $f$. In this case, it can be shown that the convergence rate of both the embedding system and the inclusion function are the same and equal to $\alpha = \infty$. Given a Lipschitz and monotone inclusion function $\OF$ for $f$, one can show that the convergence rate of the embedding system~\eqref{eq:ct-embedingsystem} satisfies $\alpha\ge 1$~\cite[Theorem 4.1]{REM-RBK-MJC:09}. However, in general, the inclusion function $\OF$ does not reveal more information about the convergence rate of the embedding system~\eqref{eq:ct-embedingsystem}. Even when the convergent rate of the inclusion function $\OF$ is $\alpha = \infty$, i.e., the embedding system is given by equations~\eqref{eq:tightembedding-c}, one might get a convergence rate of $\alpha = 1$ for the embedding system. \change{Thus, even the reachable set estimates from the embedding system associated to the minimal inclusion function can incur over-approximation errors proportional to the size of uncertainty set.}

\section{Interval reachability of learning-based systems}
\label{sec:lbs}
In this section, we develop a computationally efficient framework for studying reachability of neural network controlled systems. The key idea of our approach is to employ the framework of Section~\ref{sec:IA-d} and to capture the interactions between the open-loop system and the neural network controller using a suitable closed-loop embedding system. However, finding a closed-loop embedding system (or, equivalently, computing an inclusion function for the closed-loop vector field) using the approaches in Section~\ref{sec:IA} is often not tractable, as the neural networks  can have a large number of parameters. Instead, we propose to use a compositional approach %
based on an inclusion function for the open-loop dynamics and an inclusion function for the neural network controller.
We develop two different approaches for designing the closed-loop inclusion function from the open-loop vector field $f$ and the localized inclusion function of the neural network. 

First, we propose an %
 \emph{interconnection-based} approach where %
the key idea is to consider the closed-loop embedding system as the feedback interconnection of the open-loop embedding system and the neural network inclusion function. 
This method, as elaborated below, can capture some of the stabilizing effects of the neural network controller by approximating the input-output behavior of the neural network separately on the faces of a hyperrectangle. Thus, we consider this approach a semi-input-output approach. The advantages of the interconnection-based approach are its computational speed and that it is amenable to any open-loop inclusion function. %

Second, we propose an \emph{interaction-based} approach where  %
the key idea is to use the Jacobian-based cornered inclusion function (developed in Proposition~\ref{thm:jacobian-basedIF}) to more fully capture the interaction between the neural network controller and the system. The advantages of this approach are reduced conservatism compared to the interconnection-based inclusion function while retaining much of the computational efficiencies of interval reachability analysis.

Our methods are applicable to general nonlinear systems with neural network controllers. Nonetheless, the essence of our approaches, and their advantages over a naive input-output approach, are evident even in the simple setting of a linear system with a linear feedback controller, as shown in the following illustrative example.

\begin{example}[Invariant intervals]
Consider the control system
    $\dot{x} = Ax + Bu$
where $x\in \real^2$ and $A = \left[\begin{smallmatrix} -2 & 1\\ 1 & -2\end{smallmatrix}\right]$ and $B = \left[\begin{smallmatrix} 0\\1\end{smallmatrix}\right]$. Note that $A$ is Hurwitz %
so the origin is the globally asymptotically stable equilibrium point of the open-loop system when $u\equiv 0$. %
Consider a state feedback controller $u=\pi(x) = Kx = k_1x_1+k_2x_2$ with $k_1=k_2=-3$ designed to, \emph{e.g.}, achieve faster convergence to the origin. As a special case of reachability analysis, our goal is to characterize forward invariant intervals of the closed-loop system around the origin. We compare three different approaches for finding invariant intervals. 

\paragraph*{Naive input-output approach} this approach decouples the system from the controller and only uses knowledge of the output range of the controller. In particular, this approach seeks an interval $S = [-\xi,\xi]$, $\xi\in\mathbb{R}^2_{\geq 0}$ such that $S$ is robustly forward invariant for the open loop system %
under any $u\in \pi(S)$. However, %
such a robustly forward invariant set does not exist. By contradiction, suppose $S=[-\xi,\xi]$ is forward invariant and consider the point $x=\left[\begin{smallmatrix} 0\\ \xi_2\end{smallmatrix}\right]$ on the boundary of this set and 
pick $u=\pi(-x)= 3\xi_2\in \pi(S)$, for which  $\dot{x}_2 = -2\xi_2 + u = \xi_2>0$ and the resulting trajectory immediately leaves $S$.  %
Despite the stability of the open-loop system and the stabilizing effect of the controller, this approach fails to capture the stability of the closed-loop system.

\paragraph*{Interconnection-based approach} this approach also decouples the system from the controller and only assumes knowledge of the output of the controller to characterize forward invariant intervals $S=[-\xi,\xi]$, for some $\xi\in\mathbb{R}^2_{\geq 0}$. However, it considers the interconnection between the system and the controller on the edges of the interval $S$. Let $S_1^{-},S_1^{+},S_2^{-},S_2^{+}$ denote the edges of $S$ defined by $S_i^{\pm} = \{x\in S\mid x_i = \pm \xi_i\}$ for $i\in \{1,2\}$. Then a sufficient condition for forward invariance of $S$ is that $\dot{x}_i\ge 0$ for every $x\in S^-_i$, and $\dot{x}_i\le 0$ for every $x\in S^+_i$. Note that, if $\xi_1/\xi_2\ge \frac{1}{2}$, then we have $\dot{x}_1\le 0$ on $S_1^+$ and we have $\dot{x}_1\ge 0$ on $S_1^-$. On $S^{+}_2$, we have $\dot{x}_2 \in [-\xi_1,\xi_1] - 2\xi_2 + u$ and $u \in  -3[-\xi_1,\xi_1] - 3 \xi_2$. This implies that $\dot{x}_2 \in  4[-\xi_1,\xi_1] - 5\xi_2$, and, therefore, we need $\xi_1/\xi_2 \le \frac{5}{4}$ for $\dot{x}_2\leq 0$. A similar condition is required to ensure $\dot{x}_2\ge 0$ on $S^{-}_2$. This implies that, using this interconnection-based approach, we  certify  $S = [-\xi,\xi]$ is forward invariant for any $\xi\in \real^2_{>0}$ satisfying $\frac{1}{2}\le \xi_1/\xi_2\le \frac{5}{4}$. 

\paragraph*{Interaction-based approach} this approach uses the knowledge of the controller functional dependency $\pi(x)=Kx$ and the open-loop dynamics
to obtain the closed-loop system 
    $\dot{x} = (A + BK) x = \begin{bsmallmatrix}
        -2 & 1 \\ -2 & -5
    \end{bsmallmatrix}x. $
One can show that the interval $S = [-\xi,\xi]$ is forward invariant for this closed-loop system %
for any $\xi\in \real^2_{\ge 0}$ satisfying $\frac{1}{2}\le \xi_1/\xi_2\le \frac{5}{2}$. 

\end{example}

Now, we develop a mathematical framework for reachability analysis of the learning-based closed-loop system~\eqref{eq:closedloop-c}. Given an open-loop system of the form~\eqref{eq:dynamics-c}, we use approaches developed in Section~\ref{subsec:compute} to construct an inclusion function $\OF^{\mathrm{o}} = \left[\begin{smallmatrix}\underline{\OF}^{\mathrm{o}}\\ \overline{\OF}^{\mathrm{o}}\end{smallmatrix}\right]:\mathcal{T}^{2n}_{\ge 0}\times \mathcal{T}^{2p}_{\ge 0}\times \mathcal{T}^{2q}_{\ge 0}\to \mathcal{T}^{2n}_{\ge 0}$ for the open-loop vector field $f$. We also assume that we have access to a neural network inclusion function $\ON$ as developed in Section~\ref{sec:nninclusion}. We construct two classes of inclusion functions for the closed-loop vector field $f^c$ in~\eqref{eq:closedloop-c}.

\subsection{Interconnection-based approach}

In the first approach, which we refer to as interconnection-based approach, we obtain the closed-loop inclusion function from a feedback interconnection of the open-loop embedding system and the neural network inclusion function. 

 \begin{theorem}[Closed-loop interconnection-based inclusion function]\label{thm:cl-comp}
 Consider the open-loop system~\eqref{eq:plant-c} with an inclusion function $\OF^{\mathrm{o}}$ and a neural network controller of the form~\eqref{eq:NN} with a $[y,\widehat{y}]$-localized inclusion function $\ON_{[y,\widehat{y}]} = \left[\begin{smallmatrix}\underline{\ON}_{[y,\widehat{y}]}\\ \overline{\ON}_{[y,\widehat{y}]}\end{smallmatrix}\right]$. Then, the map $\OF_{[y,\widehat{y}]}^{\mathrm{con}}:\mathcal{T}^{2n}_{\ge 0}\times \mathcal{T}^{2q}_{\ge 0}\to \mathcal{T}^{2n}_{\ge 0}$ defined by
 \begin{align}\label{eq:int-inclusion}
    \underline{\OF}_{[y,\widehat{y}]}^{\mathrm{con}}(x,\widehat{x},{w},\widehat{w})
     & = \underline{\OF}^{\mathrm{o}}(x,\widehat{x}, {\xi},\widehat{\xi},{w},\widehat{w})\nonumber\\
     \overline{\OF}_{[y,\widehat{y}]}^{\mathrm{con}}(x,\widehat{x},{w},\widehat{w}) & = \overline{\OF}^{\mathrm{o}}(x,\widehat{x}, {\eta},\widehat{\eta},{w},\widehat{w})
\end{align}
is a $[y,\widehat{y}]\times \real^q$-localized inclusion function for the closed-loop vector field $f^{\mathrm{c}}$, where 
\begin{align*}
   {\xi} & =  \underline{\ON}_{[y,\widehat{y}]}(x,\widehat{x}),\quad \widehat{\xi} = \overline{\ON}_{[y,\widehat{y}]}(x,\widehat{x}),\\
   {\eta} &= \overline{\ON}_{[{y},\widehat{y}]}(x,\widehat{x}), \quad 
   \widehat{\eta} = \underline{\ON}_{[{y},\widehat{y}]}(x,\widehat{x}).
\end{align*}
\end{theorem}
\begin{proof}
Let $z\in [x,\widehat{x}]\subseteq [y,\widehat{y}]$ and $v\in [w,\widehat{w}]$. Since $\ON_{[y,\widehat{y}]} = \left[\begin{smallmatrix}\underline{\ON}_{[y,\widehat{y}]}\\ \overline{\ON}_{[y,\widehat{y}]}\end{smallmatrix}\right]$ is a $[y,\widehat{y}]$-localized inclusion function for the neural network controller $N$, we have $\underline{\ON}_{[y,\widehat{y}]}(x,\widehat{x}) \le N(z) \le \overline{\ON}_{[y,\widehat{y}]}(x,\widehat{x})$.
Therefore, we can compute 
$\underline{\OF}_{[y,\widehat{y}]}^{\mathrm{con}}(x,\widehat{x},{w},\widehat{w})
      = \underline{\OF}^{\mathrm{o}}(x,\widehat{x}, {\xi},\widehat{\xi},{w},\widehat{w})    \le f(z,N(z),v) = f^c(z,v),$ 
where the second inequality holds because $\OF^{\mathrm{o}}$ is an inclusion function for the open-loop vector field $f$. Similarly, one can show that $\overline{\OF}_{[y,\widehat{y}]}^{\mathrm{con}}(x,\widehat{x},{w},\widehat{w})\ge f^c(z,v)$. This implies that ${\OF}_{[y,\widehat{y}]}^{\mathrm{con}}$ is a $[y,\widehat{y}]$-localized inclusion function for $f^c$. 
\end{proof}

\subsection{Interaction-based approach}

The second approach, which we refer to as interaction-based approach, uses a Jacobian-based cornered inclusion function for the open-loop system and an affine neural network bound to capture the first order interaction between the neural network controller and the system.

 \begin{theorem}[%
 Closed-loop interaction-based inclusion function]\label{thm:jacobian-based}

Consider the closed-loop system~\eqref{eq:closedloop-c} and assume the open-loop vector field $f$ is continuously differentiable and the neural network controller~\eqref{eq:NN} satisfies affine bounds of the form~\eqref{eq:crown} on the interval $[y,\widehat{y}]$. Assume that, for every $z\in [x,\widehat{x}]\subseteq [y,\widehat{y}]$, every $\xi\in [u,\widehat{u}]$, and every $\eta\in [w,\widehat{w}]$,
\begin{align}\label{eq:boundsonf}
    D_z f(z,\xi,\eta) &\in [\underline{J}_{[x,\widehat{x}]},\overline{J}_{[x,\widehat{x}]}],\nonumber\\
    D_\xi f(z,\xi,\eta) &\in [\underline{J}_{[u,\widehat{u}]},\overline{J}_{[u,\widehat{u}]}], \nonumber\\ D_\eta f(z,\xi,\eta) &\in [\underline{J}_{[w,\widehat{w}]},\overline{J}_{[w,\widehat{w}]}].
    \end{align} 
 Then, the map $\OF_{[y,\widehat{y}]}^{\mathrm{act}}: \mathcal{T}^{2n}_{\ge 0}\times \mathcal{T}^{2q}_{\ge 0}\to \mathcal{T}^{2n}_{\ge 0}$ defined by
\begin{align}\label{eq:jac-inclusion}
\hspace{-0.4cm}\OF_{[y,\widehat{y}]}^{\mathrm{act}}(x,\widehat{x},w,\widehat{w}) 
 = \left[\begin{smallmatrix}[\underline{H}]^+ - \underline{J}_{[x,\widehat{x}]}  & [\underline{H}]^- \\ [\overline{H}]^- - \overline{J}_{[x,\widehat{x}]}  & [\overline{H}]^+\end{smallmatrix}\right]\left[\begin{smallmatrix}x\\ \widehat{x} \end{smallmatrix}\right] + L\left[\begin{smallmatrix}w\\ \widehat{w}\end{smallmatrix}\right] + Q,
\end{align} 
 is a $[y,\widehat{y}]\times \real^q$-localized inclusion function for the closed-loop vector field $f^c$, where
\begin{align*}
    \underline{H} &= \underline{J}_{[x,\widehat{x}]} + [\underline{J}_{[u,\widehat{u}]}]^+\underline{C}(y,\widehat{y}) +[\underline{J}_{[u,\widehat{u}]}]^-\overline{C}(y,\widehat{y})\\
    \overline{H} &= \overline{J}_{[x,\widehat{x}]} + [\overline{J}_{[u,\widehat{u}]}]^+\overline{C}(y,\widehat{y}) +[\overline{J}_{[u,\widehat{u}]}]^-\underline{C}(y,\widehat{y}), \\
    L & = \left[\begin{smallmatrix}-[\underline{J}_{[w,\widehat{w}]}]^- & [\underline{J}_{[w,\widehat{w}]}]^-\\ -[\overline{J}_{[w,\widehat{w}]}]^+ & [\overline{J}_{[w,\widehat{w}]}]^+\end{smallmatrix}\right], \\ Q & = \left[\begin{smallmatrix}  - \underline{J}_{[u,\widehat{u}]} u  +  [\underline{J}_{[u,\widehat{u}]}]^+\underline{d}(y,\widehat{y}) + [\underline{J}_{[u,\widehat{u}]}]^-\overline{d}(y,\widehat{y}) + f(x,u,w)\\ - \overline{J}_{[u,\widehat{u}]} u   +  [\overline{J}_{[u,\widehat{u}]}]^+\overline{d}(y,\widehat{y}) + [\overline{J}_{[u,\widehat{u}]}]^-\underline{d}(y,\widehat{y}) + f(x,u,w) \end{smallmatrix}\right]. 
\end{align*}   
 \end{theorem}
\begin{proof}
Let $z\in [x,\widehat{x}]$ and $\eta\in [w,\widehat{w}]$. Since $f$ is continuously differentiable, the fundamental theorem of Calculus gives
    \begin{align*}
        f^c(z,\eta) &= f(z,N(z),\eta) = f(x,u,w) + M(z,u,w) (z-x) \\ &  + P(z,u,w) (N(z)-u) + R(z,\eta,w)(\eta-w), 
    \end{align*}
    where $M(z,u,w) = \int_{0}^{1} D_xf (\tau z +(1-\tau)x,u,w)d\tau$, $P(z,u,w) = \int_{0}^{1} D_uf (z,\tau N(z) + (1-\tau) u,w)d\tau$, and $R(z,\eta,w) = \int_{0}^{1} D_wf (z,N(z),\tau \eta + (1-\tau)w)d\tau$.     
Using the bounds~\eqref{eq:boundsonf}, we get $M(z,u,w)  \ge  \underline{J}_{[x,\widehat{x}]}$, $P(z,u,w)  \ge \underline{J}_{[u,\widehat{u}]}$, and $R(z,\eta,w) \ge \underline{J}_{[w,\widehat{w}]}$. Moreover, we have $N(z)-u\ge \vect{0}_p$. Therefore, using the affine bounds~\eqref{eq:crown} for the neural network,
    $P(z,u,w)(N(z)-u)  \ge [P(z,u,w)]^+(\underline{C} (y,\widehat{y}) z +\underline{d}(y,\widehat{y})-u)   +[P(z,u,w)]^-(\overline{C}(y,\widehat{y})z + \overline{d}(y,\widehat{y})-u).$ 
    Using the fact that $z-x\ge \vect{0}_n$ and $
    \eta-w\ge \vect{0}_q$, we have 
     \begin{align*}
        f^c(z,\eta) & \ge f(x,u,w) + M(z,u,w) (z-x) \\ & + 
        [P(z,u,w)]^+(\underline{C}(y,\widehat{y})z + \overline{d}(y,\widehat{y})-u)\\ & + [P(z,u,w)]^{-} (\overline{C}(y,\widehat{y})z + \overline{d}(y,\widehat{y})-u) \\ & + R(z,\eta,w)(\eta-w)
        \\& \ge 
        \underline{H}z -\underline{J}_{[x,\widehat{x}]} x +  \underline{J}_{[w,\widehat{w}]} \eta - \underline{J}_{[w,\widehat{w}]} w \\ & + [\underline{J}_{[w,\widehat{w}]}]^{+}\underline{d}(y,\widehat{y}) +[\underline{J}_{[w,\widehat{w}]}]^{-}\overline{d}(y,\widehat{y}) +  f(x,u,w),
    \end{align*}
    where the second inequality follows from the Jacobian bounds. Now, we note that $x\le z\le \widehat{x}$ and $w\le \eta \le \widehat{w}$ and therefore
      $  \underline{H}z \ge [\underline{H}]^+ x + [\underline{H}]^+ \widehat{x}$ and $
    \underline{J}_{[w,\widehat{w}]}\eta \ge [\underline{J}_{[w,\widehat{w}]}]^+ w + [\underline{J}_{[w,\widehat{w}]}]^- \widehat{w}.$ 
    As a result, we get 
    \begin{align*}
         f^c(z,\eta) &\ge [\underline{H}]^{+}x + [\underline{H}]^{-}\widehat{x} - \underline{J}_{[x,\widehat{x}]} x - [\underline{J}_{[w,\widehat{w}]}]^{-} w \\ & + [\underline{J}_{[w,\widehat{w}]}]^-\widehat{w}  +   [\underline{J}_{[w,\widehat{w}]}]^{+}\underline{d}(y,\widehat{y}) +[\underline{J}_{[w,\widehat{w}]}]^{-}\overline{d}(y,\widehat{y}) \\ & +  f(x,u,w).
    \end{align*}
    This implies that $f^c(z,\eta) \ge \underline{\OF}_{[y,\widehat{y}]}^{\mathrm{con}}(x,\widehat{x},w,\widehat{w})$. Similarly, one can show that $f^c(z,\eta) \le \overline{\OF}_{[y,\widehat{y}]}^{\mathrm{con}}(x,\widehat{x},w,\widehat{w})$. %
\end{proof}

\begin{remark}\label{rem:important}
\;
\begin{enumerate}
\item \textit{(Jacobian-based cornered inclusion function):} Theorem~\ref{thm:jacobian-based} provides a localized inclusion function for the closed-loop vector field $f^c$ using the Jacobian-based cornered inclusion function (as in Proposition~\ref{thm:jacobian-basedIF}) for the corner point $(x,u,w)$. First, one can obtain different but similar localized inclusion functions using Taylor expansion of $f$ around other corner points such as $(\widehat{x},{u},w)$, $(\widehat{x},\widehat{u},w)$, etc. Our numerical results show that, in some cases, considering a few corner points and combining the resulting localized inclusion functions using Proposition \ref{prop:intersect} significantly reduces conservatism at low computational expense.
Second, as it is clear from the proof of Theorem~\ref{thm:jacobian-based}, the Jacobian-based cornered inclusion function is particularly well-suited for integrating the neural network bounds into the closed-loop inclusion function and for studying the interaction between the system and the neural network controller. 

\item\label{p2:important} \textit{(Mixed Jacobian-based inclusion functions):} Theorem~\ref{thm:jacobian-based} uses the Jacobian-based cornered inclusion function (as in Proposition~\ref{thm:jacobian-basedIF}) for the open-loop vector field $f$. Alternatively, one can use the mixed Jacobian-based cornered inclusion function (as in Proposition~\ref{thm:mixed}) for the open-loop vector $f$ and obtain a different class of inclusion functions for $f^c$. This class of closed-loop inclusion functions can be obtained from~\eqref{eq:jac-inclusion} by replacing $\underline{J}_{[x,\widehat{x}]},\overline{J}_{[x,\widehat{x}]},\underline{J}_{[u,\widehat{u}]},\overline{J}_{[u,\widehat{u}]},\underline{J}_{[w,\widehat{w}]}$, and $\overline{J}_{[w,\widehat{w}]}$ with their counterparts $\underline{M}_{[x,\widehat{x}]},\overline{M}_{[x,\widehat{x}]},\underline{M}_{[u,\widehat{u}]},\overline{M}_{[u,\widehat{u}]},\underline{M}_{[w,\widehat{w}]}$, and $\overline{M}_{[w,\widehat{w}]}$ as defined in Proposition~\ref{thm:mixed} and is demonstrated in the numerical experiments below. Our numerical results show that, in most cases, mixed Jacobian-based cornered inclusion functions significantly reduces the conservatism of the over-approximation as compared to their non-mixed counterparts.

\item \textit{(Comparison):} The interconnection-based approach provides a general and computationally efficient framework for constructing closed-loop inclusion functions. This approach can be applied using any inclusion function for the open-loop system and any inclusion function of the form~\eqref{eq:NNverifier} for the neural network. The computational efficiency of the interconnection-based approach is a direct consequence of its compositional structure. %
On the other hand, the interaction-based approach fully captures the first-order interaction between the system and the neural networks using a Jacobian-based decomposition. However, this approach requires affine bounds of the form~\eqref{eq:crown} for the neural network and  Jacobian bounds of the form~\eqref{eq:boundsonf} for the open-loop system. As a result, computing the closed-loop inclusion functions using the interaction-based approach is generally more computationally expensive.
\end{enumerate}

\end{remark}

When the open-loop  system~\eqref{eq:plant-c} is linear, one can significantly simplify the expression for the closed-loop inclusion functions obtained from the interconnection-based and the interaction-based approaches. 

\begin{corollary}[Linear open-loop systems]\label{thm:linear}
Consider the closed-loop system~\eqref{eq:closedloop-c} and assume that the open-loop vector field $f$ is linear with the form $ f(x,u,w) = A x + B u + D w$, where $A\in \real^{n\times n}$, $B\in \real^{n\times p}$ and $D\in \real^{n\times q}$.  Then the following statements hold:
     \begin{enumerate}
         \item\label{p1:linear} using the minimal inclusion function for the open-loop vector field $f$ and affine neural network inclusion functions of the form~\eqref{eq:crownif}, the closed-loop inclusion function $\OF^{\mathrm{con}}_{[y,\widehat{y}]}$ defined in~\eqref{eq:int-inclusion} is given by:
         {\color{black}\begin{align}\label{eq:linear-inclusion-int}
            \hspace{-0.4cm} \OF_{[y,\widehat{y}]}^{\mathrm{con}} &= \left[\begin{smallmatrix} [A]^+ + \left[[B]^{+}\underline{C} + [B]^-\overline{C}\right]^{+} & [A]^- + \left[[B]^{+}\underline{C} + [B]^-\overline{C}\right]^{-} \nonumber \\  [A]^- + \left[[B]^{+}\overline{C} + [B]^-\underline{C}\right]^{-}& [A]^+ + \left[[B]^{+}\overline{C} + [B]^-\underline{C}\right]^{+} \end{smallmatrix}\right]\left[\begin{smallmatrix} x \\ \widehat{x} \end{smallmatrix}\right] \\ & + \left[\begin{smallmatrix} [D]^+ & [D]^- \\ [D]^- & [D]^+ \end{smallmatrix}\right]\left[\begin{smallmatrix} w \\ \widehat{w} \end{smallmatrix}\right] + \left[\begin{smallmatrix} [B]^+\underline{d} + [B]^-\overline{d} \\  [B]^-\underline{d} + [B]^+\overline{d} \end{smallmatrix}\right].
         \end{align}}

         \item\label{p2:linear} the inclusion function $\OF_{[y,\widehat{y}]}^{\mathrm{act}}$ defined in~\eqref{eq:jac-inclusion} is given by
{\color{black}\begin{align}\label{eq:linear-inclusion-jac}
             \OF_{[y,\widehat{y}]}^{\mathrm{act}} &= \left[\begin{smallmatrix} [A+ B^{+}\underline{C} + B^-\overline{C}]^{+} & [A+ B^{+}\underline{C} + B^-\overline{C}]^{-} \\  [A+ B^{+}\overline{C} + B^-\underline{C}]^{-}& [A+ B^{+}\overline{C} + B^-\underline{C}]^{+} \end{smallmatrix}\right]\left[\begin{smallmatrix} x \\ \widehat{x} \end{smallmatrix}\right] \nonumber\\ & + \left[\begin{smallmatrix} [D]^+ & [D]^- \\ [D]^- & [D]^+\end{smallmatrix}\right]\left[\begin{smallmatrix} w \\ \widehat{w} \end{smallmatrix}\right] + \left[\begin{smallmatrix} [B]^+\underline{d} + [B]^-\overline{d} \\  [B]^-\underline{d} + [B]^+\overline{d} \end{smallmatrix}\right].
         \end{align}}
         \vspace{-1em}
     \end{enumerate} 
\end{corollary}
    \begin{proof}
         Regarding part~\ref{p1:linear}, it is easy to see that the minimal inclusion function of the linear open-loop vector field $f$ is
         \begin{align*}
             \OF^{\mathrm{o}}(x,\widehat{x},u,\widehat{u},w,\widehat{w}) & = \left[\begin{smallmatrix} [A]^+ & [A]^-\\  [A]^-& [A]^{+}\end{smallmatrix}\right] 
\left[\begin{smallmatrix} x\\ \widehat{x}\end{smallmatrix}\right] + \left[\begin{smallmatrix} [B]^+ & [B]^-\\  [B]^-& [B]^{+}\end{smallmatrix}\right] 
\left[\begin{smallmatrix} u\\ \widehat{u}\end{smallmatrix}\right] \\ & + \left[\begin{smallmatrix} [D]^+ & [D]^-\\  [D]^-& [D]^{+}\end{smallmatrix}\right] 
\left[\begin{smallmatrix} w\\ \widehat{w}\end{smallmatrix}\right].    
         \end{align*}
         One can then replace the above minimal inclusion function and the affine bounds~\eqref{eq:crownif} into equation~\eqref{eq:int-inclusion} and use Theorem~\ref{thm:cl-comp} to obtain the result. Regarding part~\ref{p2:linear}, we use the bounds $\underline{J}_{[x,\widehat{x}]} = \overline{J}_{[x,\widehat{x}]} = A$, $\underline{J}_{[u,\widehat{u}]} = \overline{J}_{[u,\widehat{u}]} = B$, and $\underline{J}_{[w,\widehat{w}]} = \overline{J}_{[w,\widehat{w}]} = D$ in Theorem~\ref{thm:jacobian-based}. Thus, using the notation of Theorem~\ref{thm:jacobian-based}, we compute 
         \begin{align*}
             \underline{H} &= A + [B]^+\underline{C}(y,\widehat{y}) + [B]^- \overline{C}(y,\widehat{y}),\\ 
             \overline{H} &=  A + [B]^+\overline{C}(y,\widehat{y}) + [B]^- \underline{C}(y,\widehat{y}), \\
             L &= \left[\begin{smallmatrix}  -[D]^- & [D]^- \\
             -[D]^+ & [D]^+\end{smallmatrix}\right],\qquad  
             Q = \left[\begin{smallmatrix}
             Ax + D w  + [B]^+\underline{d}(y,\widehat{y}) + [B]^- \overline{d}(y,\widehat{y})\\
             Ax + D w + [B]^+\overline{d}(y,\widehat{y}) + [B]^- \underline{d}(y,\widehat{y})
             \end{smallmatrix}\right].
         \end{align*}
         The result then follows by replacing the above terms into equation~\eqref{eq:jac-inclusion} and using Theorem~\ref{thm:jacobian-based}.
     \end{proof}

\begin{remark}[The role of interactions] The closed-form expressions~\eqref{eq:linear-inclusion-int} and~\eqref{eq:linear-inclusion-jac} demonstrate the difference between the interconnection-based and the interaction-based approach. In both cases, one can interpret the term $A$ as the effect of open-loop dynamics and the terms $B^{+}\underline{C} + B^{-}\overline{C}$ and $B^{+}\overline{C} + B^{-}\underline{C}$ as the effect of the neural network controller. The interconnection-based approach considers the interconnection of the cooperative and competitive effect of the open-loop system and the neural network controller, separately. This leads to the terms  $[A]^{\pm} + [B^{+}\underline{C} + B^{-}\overline{C}]^{\pm}$ and  $[A]^{\pm} + [B^{+}\overline{C} + B^{-}\underline{C}]^{\pm}$. The interaction-based approach \change{instead} considers the interaction between the open-loop system and the neural network controller via the terms $A + B^{+}\underline{C} + B^{-}\overline{C}$ and $A+B^{+}\overline{C} + B^{-}\underline{C}$ and then the cooperative and competitive effect of these interactions are separated. 
\end{remark}

\subsection{Interval Reachability of Closed-loop Systems}
Our culminating conclusion is that, by using the interconnection-based and the interaction-based inclusion functions, 
we obtain computationally efficient over-approximations of reachable sets of the closed-loop system. %

\begin{theorem}[Reachability of closed-loop system]\label{thm:cl-comp-c}
Let $\mathrm{c}\in \{\mathrm{con},\mathrm{act}\}$, the disturbance set $\mathcal{W}=[\underline{w},\overline{w}]$, and the initial set $\mathcal{X}_0=[\underline{x}_0,\overline{x}_0]$. For the closed-loop dynamical system~\eqref{eq:closedloop-c} with $t\mapsto \left[\begin{smallmatrix}\underline{x}^{\mathrm{c}}(t)\\ \overline{x}^{\mathrm{c}}(t)\end{smallmatrix}\right]$ being the trajectory of the embedding system
\begin{align}\label{eq:ct-good}
\dot{x}_i & = \left(\underline{\OF}_{[x,\widehat{x}]}^{\mathrm{c}}(x,\widehat{x}_{[i:x]},\ulw,\olw)\right)_i,\nonumber\\
\dot{\widehat{x}}_i & =\left(\overline{\OF}_{[x,\widehat{x}]}^{\mathrm{c}}(x_{[i:\widehat{x}]},\widehat{x},\olw,\ulw)\right)_i,
\end{align}
for every $i\in\{1,\ldots,n\}$, starting from $\left[\begin{smallmatrix}\underline{x}_0\\ \overline{x}_0\end{smallmatrix}\right]$, the following statements hold\change{, for every $t\ge 0$}:
 \begin{enumerate}
     \item\label{p1:cl-comp-c} $\mathcal{R}_{f^{\mathrm{c}}}(t,[\underline{x}_0,\overline{x}_0], [\underline{w},\overline{w}])\subseteq [\underline{x}^{\mathrm{c}}(t),\overline{x}^{\mathrm{c}}(t)]$;
     \item\label{p2:cl-comp-c} moreover, if the open-loop system~\eqref{eq:plant-c} is linear, then
     \begin{align*}
[\underline{x}^{\mathrm{act}}(t),\overline{x}^{\mathrm{act}}(t)]\subseteq [\underline{x}^{\mathrm{con}}(t),\overline{x}^{\mathrm{con}}(t)]. 
\end{align*}
 \end{enumerate}

\end{theorem}
\begin{proof}
Regarding part~\ref{p1:cl-comp-c}, using  Theorem~\ref{thm:cl-comp} and Theorem~\ref{thm:jacobian-based}, we show that $\OF^{\mathrm{con}}_{[y,\widehat{y}]}$ and $\OF^{\mathrm{act}}_{[y,\widehat{y}]}$ are $[y,\widehat{y}]\times \real^q$-localized inclusion functions for the closed-loop vector field $f^c$. Thus, one can use the localized version of Proposition~\ref{thm:reach-c} (see Remark~\ref{rem:localized}\ref{p1:localized}) with $y=x$ and $\widehat{y}=\widehat{x}$ to obtain the result. Regarding part~\ref{p2:cl-comp-c}, for linear open-loop systems, using Theorem~\ref{thm:linear}, the embedding system associated to the interconnection-based inclusion function has the following form:
{\color{black}\begin{align}\label{eq:embedding-linear-int}
    \tfrac{d}{dt}\left[\begin{smallmatrix} x \\ \widehat{x} \end{smallmatrix}\right] &= \left[\begin{smallmatrix} [A]^{\mathrm{M}} + [B^{+}\underline{C} + B^-\overline{C}]^{\mathrm{M}} & [A]^{\mathrm{nM}} + [B^{+}\underline{C} + B^-\overline{C}]^{\mathrm{nM}} \\  [A]^{\mathrm{nM}} + [B^{+}\underline{C} + B^-\overline{C}]^{\mathrm{nM}}& [A]^{\mathrm{M}} + [B^{+}\underline{C} + B^-\overline{C}]^{\mathrm{M}} \end{smallmatrix}\right]\left[\begin{smallmatrix} x \\ \widehat{x} \end{smallmatrix}\right] \nonumber \\ & + \left[\begin{smallmatrix} [D]^+ & [D]^- \\ [D]^- & [D]^+ \end{smallmatrix}\right]\left[\begin{smallmatrix} w \\ \widehat{w} \end{smallmatrix}\right] + \left[\begin{smallmatrix} [B]^+\underline{d} + [B]^-\overline{d} \\  [B]^-\underline{d} + [B]^+\overline{d} \end{smallmatrix}\right] \nonumber\\ & = \OF^{\mathrm{con}}(x,\widehat{x},w,\widehat{w}),
\end{align}}
and the embedding system associated to the interaction-based inclusion function has the following form:
{\color{black}\begin{align}\label{eq:embedding-linear-jac}
    \tfrac{d}{dt}\left[\begin{smallmatrix} x \\ \widehat{x} \end{smallmatrix}\right] &= \left[\begin{smallmatrix} [A + [B]^{+}\underline{C} + [B]^-\overline{C}]^{\mathrm{M}} & [A + [B]^{+}\underline{C} + [B]^-\overline{C}]^{\mathrm{nM}} \\  [A + [B]^{+}\underline{C} + [B]^-\overline{C}]^{\mathrm{nM}}& [A + [B]^{+}\underline{C} + [B]^-\overline{C}]^{\mathrm{M}} \end{smallmatrix}\right]\left[\begin{smallmatrix} x \\ \widehat{x} \end{smallmatrix}\right] \nonumber \\ & + \left[\begin{smallmatrix} [D]^+ & [D]^- \\ [D]^- & [D]^+ \end{smallmatrix}\right]\left[\begin{smallmatrix} w \\ \widehat{w} \end{smallmatrix}\right] + \left[\begin{smallmatrix} [B]^+\underline{d} + [B]^-\overline{d} \\  [B]^-\underline{d} + [B]^+\overline{d} \end{smallmatrix}\right] \nonumber\\&= \OF^{\mathrm{act}}(x,\widehat{x},w,\widehat{w}).
\end{align}}
It is easy to check that both embedding systems~\eqref{eq:embedding-linear-int} and~\eqref{eq:embedding-linear-jac} are continuous-time monotone. Moreover, we have
{\color{black}\begin{align*}
 \left[A + [B]^{+}\underline{C} + [B]^-\overline{C}\right]^{\mathrm{M}} &\le [A]^{\mathrm{M}} + \left[[B]^{+}\underline{C} + [B]^-\overline{C}\right]^{\mathrm{M}},\\
 \left[A + [B]^{+}\underline{C} + [B]^-\overline{C}\right]^{\mathrm{nM}} &\ge [A]^{\mathrm{nM}} + \left[[B]^{+}\underline{C} + [B]^-\overline{C}\right]^{\mathrm{nM}}.
\end{align*}}
Therefore, for every $x\le \widehat{x}$, we have 
{\color{black}\begin{multline*}
 \left[A + [B]^{+}\underline{C} + [B]^-\overline{C}\right]^{\mathrm{M}} x + \left[A + [B]^{+}\underline{C} + [B]^-\overline{C}\right]^{\mathrm{nM}} \widehat{x} \\ \ge    ([A]^{\mathrm{M}} + [[B]^{+}\underline{C} + [B]^-\overline{C}]^{\mathrm{M}}) x \\ + ([A]^{\mathrm{nM}} + [[B]^{+}\underline{C} + [B]^-\overline{C}]^{\mathrm{nM}})\widehat{x}.
\end{multline*}}
This implies that, $\OF^{\mathrm{con}}(x,\widehat{x},w,\widehat{w})\le_{\mathrm{SE}} \OF^{\mathrm{act}}(x,\widehat{x},w,\widehat{w})$, for every  $ x \le \widehat{x}$ and every $w \le \widehat{w}$. Now, we can use the classical monotone comparison Lemma~\cite[Theorem 3.8.1]{ANM-LH-DL:08-new} to obtain $\left[\begin{smallmatrix} \underline{x}^{\mathrm{con}}(t) \\ \overline{x}^{\mathrm{con}}(t) \end{smallmatrix}\right] \le_{\mathrm{SE}} \left[\begin{smallmatrix} \underline{x}^{\mathrm{act}}(t) \\ \overline{x}^{\mathrm{act}}(t) \end{smallmatrix}\right]$ for every $t\in \real_{\ge 0}$.
As a result, we get $[\underline{x}^{\mathrm{act}}(t) ,\overline{x}^{\mathrm{act}}(t)]\subseteq [\underline{x}^{\mathrm{con}}(t) ,\overline{x}^{\mathrm{con}}(t)]$, for every $t\in \real_{\ge 0}$.
\end{proof}

 \begin{remark}
\change{The interval reachability framework developed in Sections~\ref{sec:IA-d} and~\ref{sec:lbs} has the following unique features.}
\begin{enumerate}
    \item \textit{(Computational efficiency)}: from a computational perspective, the framework presented in Theorem~\ref{thm:cl-comp-c} consists of two main ingredients: (i) the neural network verification algorithm for computing the inclusion function for the neural network (either in the form~\eqref{eq:NNverifier} or in the form~\eqref{eq:crownif}), and (ii) integration of a single trajectory of the embedding system~\eqref{eq:ct-embedingsystem}. Part (i) includes querying the neural network verification once per integration step and its runtime depends on the computational complexity of the associated algorithm. The runtime of part (ii) depends on the integration method and the form of the open-loop decomposition function $\OF$. \change{Since our approach only requires integration of one trajectory of the embedding system, it is generally computationally efficient. In Section~\ref{subsec:platooning}, we use a vehicle platooning example to demonstrate our approach's scalability to large dimensions where existing functional approaches suffer. }

    \item  \textit{(Nonlinearity of the dynamics):} For linear open-loop systems, a straightforward computation shows that Theorem~\ref{thm:cl-comp-c} with $c=\mathrm{act}$ coincides with~\cite[Lemma IV.3]{ME-GH-CS-JPH:21} with $p=\infty$ and $q=1$. \change{Our interaction-aware inclusion function generalizes this efficient approach to locally Lipschitz nonlinear systems.} 

    \change{\item\textit{(Flexibility of neural network verifiers):} The functional approaches from POLAR~\cite{CH-JF-XC-WL-QZ:22} and JuliaReach~\cite{SB-etal:19} are designed with specific neural network bounding algorithms. On the other hand, our framework can substitute arbitrary neural network interval verifiers for the interconnection-based approach, and functional linear bounds for the interaction-based approach. }
    




\end{enumerate}
\end{remark}

\section{Numerical Experiments}

We demonstrate the effectiveness of our reachability analysis, implemented as an open-source Python toolbox called \texttt{ReachMM}\footnote{\url{https://github.com/gtfactslab/ReachMM_TAC2023}}, using numerical experiments\footnote{All the experiments are performed using an AMD Ryzen 5 5600X CPU and 32 GB of RAM. Unless otherwise specified, runtimes are averaged over 100 runs, with mean and standard deviations reported.} for (i) a nonlinear bicycle model, (ii) the linear double integrator, (iii) several existing benchmarks in the literature, and (iv) a platoon of double integrator vehicles moving in a plane. %
We first briefly mention several algorithmic techniques that are implemented in the \texttt{ReachMM} toolbox and are used in several examples below to broaden the applicability and to improve the accuracy of our reachable set over-approximations. 

\paragraph*{Partitioning} %
The computational efficiency of our method makes it well-suited for partitioning to reduce compounding over-approximation error caused by the wrapping effect~\cite[Section 2.2.4]{LJ-MK-OD-EW:01}.  In Section~\ref{sec:DI}, we apply the partitioning methods proposed in~\cite{SJ-AH-SC:23,AH-SJ-SC:23} to improve  accuracy. %

\paragraph*{Redundant Variable Refinement} We accommodate a technique introduced in \cite{KS-JKS:17} to refine interval over-approximations of reachable sets using redundant state variables of the form 
$y= Ax + b$ for some $A\in\R^{m\times n}$ and $b\in\R^m$.
By augmenting the dynamical system~\eqref{eq:closedloop-c} with the additional dynamics $\dot{y}=A\dot{x}$, %
we create a new dynamical system with the new states $z = \begin{bsmallmatrix} x \\ y \end{bsmallmatrix} \in\R^{n + m}$ and the affine constraints $Mz = b$ where $M = \begin{bmatrix}
    -A & I_m 
\end{bmatrix} \in \real^{m\times (n+m)}$. Following~\cite{KS-JKS:17}, we use the interval bounds on $z$ to improve the accuracy of the interval bounds on $x$. We apply this technique in Section~\ref{sec:ARCH} for the TORA benchmark.   

\paragraph*{Zero-order hold control} In some applications, we wish to explicitly account for the practical restriction that the control $u(t)=N(x(t))$ cannot be continuously updated 
and, instead, the control must be implemented via, \emph{e.g.}, a zero-order hold strategy between sampling instances. %
It is straightforward to extend the interconnection-based approach to be able to capture the zero-order hold policy~\cite{AH-SJ-SC:23}. %
Adopting the zero-order hold policy in the interaction-based approach requires over-bounding the error and we omit this calculation for the sake of brevity. We refer to our code for details of this adaptation. We use this framework in Section~\ref{sec:ARCH}.

\paragraph*{Implementation} Our current implementation for all examples is in standard Python. For the interconnection-based approach, the inclusion function for the open-loop system is the natural inclusion functions computed using our package \texttt{npinterval}~\cite{AH-SJ-SC:23b}. For neural networks, the affine inclusion functions are obtained from CROWN~\cite{HZ-etal:18} and computed using~\texttt{autoLiRPA}~\cite{xu2020automatic}. Further performance improvements of \texttt{ReachMM} are likely possible by, \emph{e.g.}, implementing as compiled code, and are subject of ongoing and future work.

\subsection{Nonlinear bicycle model}

In the first experiment, we compare %
the interconnection-based and the interaction-based approach. %
Consider the nonlinear dynamics of a bicycle adopted from~\cite{PP-FA-BdAN-AdlF:17}:
\begin{xalignat}{2} \label{eq:vehicle}
     \dot{p_x} &= v \cos(\phi + \beta(u_2)) & \dot{\phi} &=\frac{v}{\ell_r}\sin(\beta(u_2))\nonumber\\
     \dot{p_y} &= v \sin(\phi + \beta(u_2)) & \dot{v} &= u_1
\end{xalignat} 
where $[p_x,p_y]^{\top}\in \real^2$ is the displacement of the center of mass in the $x-y$ plane, $\phi \in [-\pi,\pi)$ is the heading angle in the plane, and $v\in \real_{\geq 0}$ is the speed of the center of mass. Control input $u_1$ is the applied force, input $u_2$ is the angle of the front wheel, and $\beta(u_2) = \mathrm{arctan}\left(\frac{\ell_f}{\ell_f+\ell_r}\tan(u_2)\right)$ is the slip slide angle where the parameter $\ell_f$ ($\ell_r$) is the distance between the center of mass and the front (rear) wheel. In this example, for the sake of simplicity, we set $\ell_f=\ell_r=1$. Let $x=[p_x,p_y,\phi,v]^\top$ and $u=[u_1,u_2]^\top$. We apply the neural network controller ($4\times 100\times 100\times 2$, ReLU activations) defined in \cite{SJ-AH-SC:23}, which was trained to mimic an MPC that stabilizes the vehicle to the origin while avoiding a circular obstacle centered at $(4,4)$ with a radius of $2$. The dynamics are simulated using Euler integration with a step size of $0.125$. In Figure~\ref{fig:bicycle}, we compare the accuracy and runtime of different reachability approaches for the bicycle model. 

\paragraph*{Discussion} Figure~\ref{fig:bicycle} shows that the naive input-output approach is the fastest approach with low accuracy of over-approximation. Using the interconnection-based approach with $\OF^{\mathrm{con}}$ improves the accuracy with a slight increase in the runtime. In comparison, the interaction-based approach with the Jacobian-based cornered inclusion function $\OF^{\mathrm{act}}$ is notably slower %
due to the computation of Jacobian bounds~\eqref{eq:boundsonf} in this approach. Over short time horizons, the interaction-based approach is more accurate,  %
however, its accuracy deteriorates for longer horizons, which can be attributed to the decrease in accuracy of the Jacobian bounds~\eqref{eq:boundsonf}. As has been shown in Proposition~\ref{prop:intersect}, the accuracy of the intersection $\OF^{\mathrm{con}}\wedge\OF^{\mathrm{act}}$ is better than both $\OF^{\mathrm{con}}$ and $\OF^{\mathrm{act}}$. Finally, using the mixed Jacobian-based cornered inclusion function (cf. Remark~\ref{rem:important}\ref{p2:important}) significantly improves the accuracy of the interaction-based approach with little effect on its runtime. 

\begin{figure}
    \centering
  \includegraphics[width=1\columnwidth]{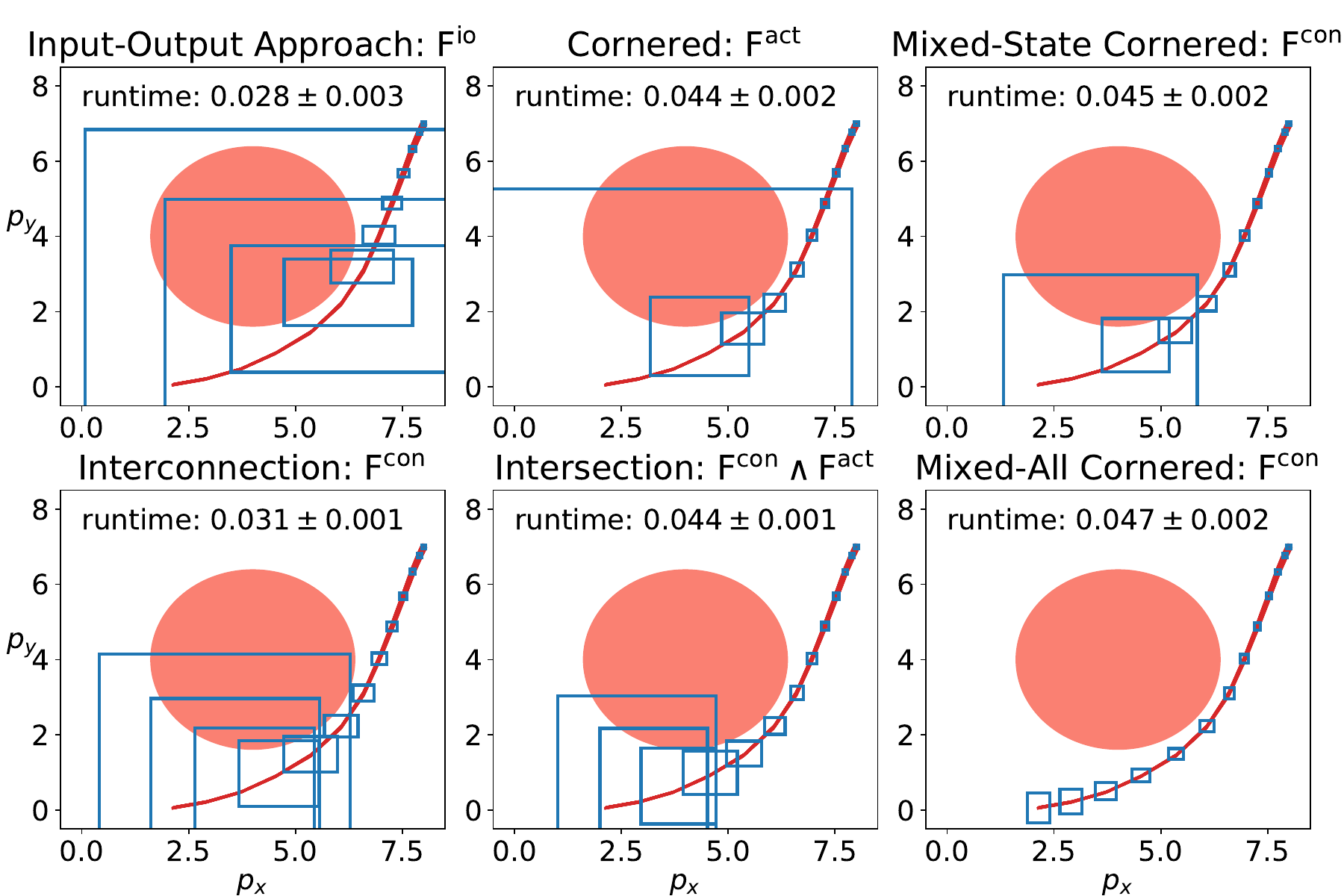}
  \vspace{-0.5cm}
    \caption{Accuracy and runtime comparison between different reachability approaches for the bicycle model~\eqref{eq:vehicle} from the initial set $[7.95,8.05]\times[6.95,7.05]\times[-2\pi/3 - 0.01, -2\pi/3 + 0.01]\times[1.99,2.01]$ \change{for the period $t\in [0,1.5]$}: (a) a naive input-output approach combining the natural inclusion function for the open-loop dynamics and the affine inclusion function for the neural network, (b) the interconnection-based approach with the closed-loop inclusion function $\OF^{\mathrm{con}}$ defined in~\eqref{eq:int-inclusion} constructed from the natural inclusion function for the open loop dynamics and affine inclusion function for the neural network, (c) the interaction-based approach with Jacobian-based cornered inclusion function $\OF^{\mathrm{act}}$ defined in~\eqref{eq:jac-inclusion}, (d) the intersection of the interconnection-based inclusion and interaction-based approach $\OF^{\mathrm{con}}\wedge \OF^{\mathrm{act}}$, (e) the interaction-based approach with \emph{mixed states} Jacobian-based cornered inclusion function defined in~\eqref{eq:jac-inclusion} and Remark~\ref{rem:important}\ref{p2:important}, and (f) the interaction-based approach with \emph{mixed states and control} Jacobian-based cornered inclusion function defined ~\eqref{eq:jac-inclusion} and Remark~\ref{rem:important}\ref{p2:important}. %
    The {\color{tab:blue} blue} boxes are hyper-rectangular over-approximation of reachable sets, $100$ simulated trajectories of the system are shown in {\color{tab:red}red}\change{, and the {\color{salmon}salmon} colored region represents a circular obstacle.} } 
    \label{fig:bicycle}
\end{figure}

\subsection{Double Integrator Model}\label{sec:DI}

In the second experiment, we focus on reachability of linear open-loop systems with neural network controllers. We study the accuracy and efficiency of the interaction-based approach~\eqref{eq:linear-inclusion-jac} for the double integrator benchmark system %
\begin{align} \label{eq:doubleintegrator}
    x(t+1)=\begin{bmatrix} 1 & 1 \\ 0 & 1 \end{bmatrix} x(t) + \begin{bmatrix} 0.5 \\ 1 \end{bmatrix} u(t).
\end{align} 
For this example, we use the discrete-time version of our framework as derived in~\cite[Section VII.B]{AH-SJ-SC:23}. 
We apply the neural network controller ($2\times 10\times 5\times 1$, ReLU activations) defined in 
\cite{ME-GH-CS-JPH:21}. \change{We consider the interaction-based approach with the closed-loop inclusion function $\OF^{\mathrm{act}}$ defined in~\eqref{eq:linear-inclusion-jac}, with the contraction-guided adaptive partitioning algorithm in~\cite{AH-SJ-SC:23} to improve its accuracy.} Additionally, we compare our proposed ReachMM to state-of-the-art algorithms for linear discrete-time systems: ReachLP~\cite{ME-GH-CS-JPH:21} with uniform partitioning (ReachLP-Unif) and greedy simulation guided partitioning (ReachLP-GSG), and ReachLipBnB \cite{TE-SS-MF:22} (branch-and-bound using LipSDP \cite{MF-MM-GJP:22}). Each algorithm is run with two different sets of hyper-parameters, aiming to compare their performances across various regimes. The setup for ReachMM is $(\varepsilon,\,D_p,\,D_\text{N})$\footnote{\change{$D_p$ defines the depth of the partition tree, $D_\text{N}$ defines how the depth of the partitions on the neural network, and $\varepsilon$ describes the width at which to adaptively partition an interval. See~\cite{AH-SJ-SC:23} for the full details.}}; ReachLP-Unif is $\#$ initial partitions, ReachLP-GSG is $\#$ of total propogator calls, ReachLipBnB is $\varepsilon$. The results are \change{illustrated} in Figure~\ref{fig:doubleintegrator} and the runtime comparisons are \change{provided} in 
Table~\ref{tab:DI_table}.

\begin{figure}
    \centering
    \includegraphics[width=\columnwidth]{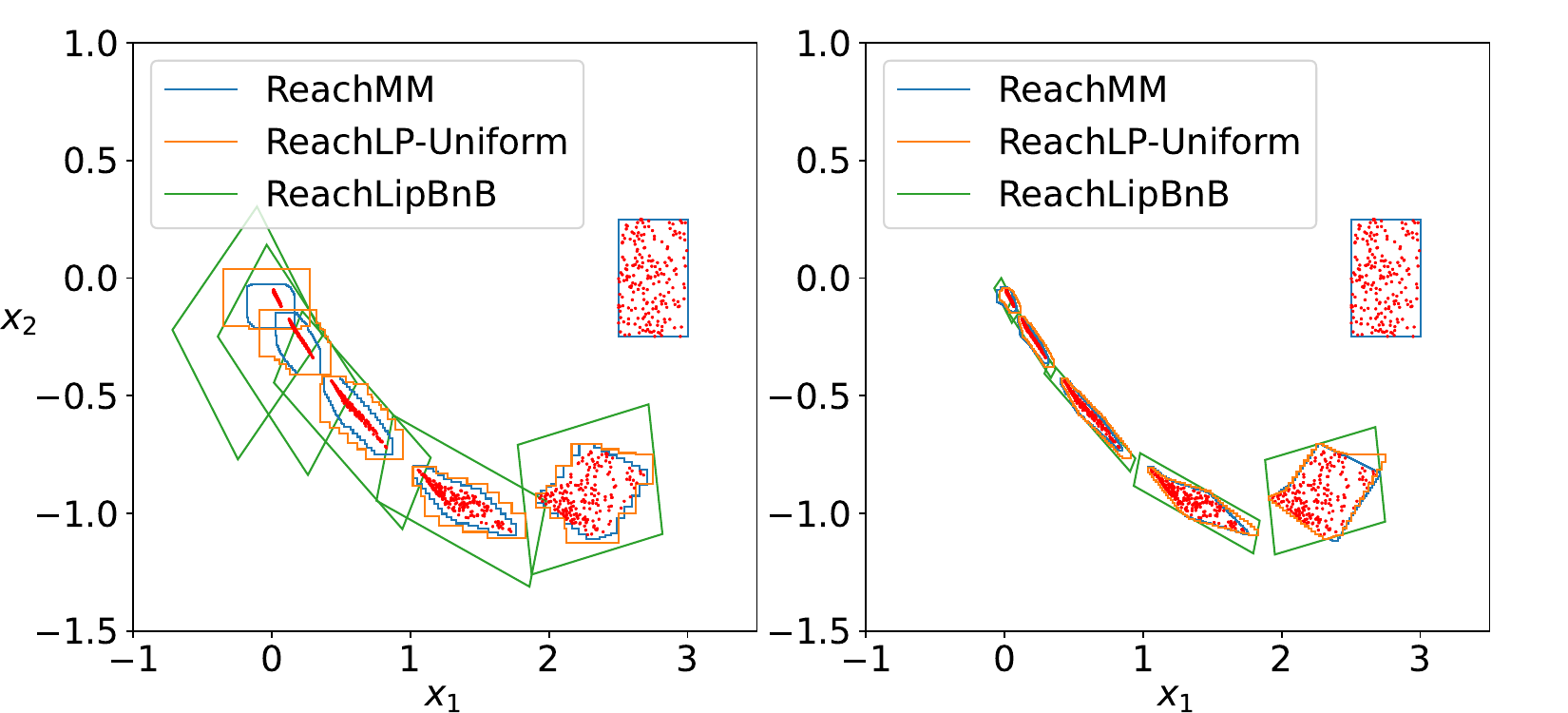}
    \vspace{-0.5cm}
    \caption{The over-approximated reachable sets of the closed-loop double integrator model~\eqref{eq:doubleintegrator} are compared for three different algorithms on two different runtime regimes, for the initial set $[2.5,3]\times[-0.25,0.25]$ and final time $T=5$. The experiment setup and performances are reported in Table~\ref{tab:DI_table}. 200 true trajectories are shown in red. The horizontal axis is $x_1$ and the vertical axis is $x_2$.}
    \label{fig:doubleintegrator}
\end{figure}

\begin{table}[h!]
    \centering
    \begin{tabular}{||c c c c||}
        \hline
        Method & Setup & Runtime (s) & Area \\
        \hline\hline
{\footnotesize ReachMM} & $(0.1$, $3$, $1)$  & $\mathbf{0.103\pm0.003}$ & $\mathbf{6.2\cdot10^{-2}}$ \\
 & $(0.05$, $6$, $2)$ & $\mathbf{1.762\pm0.026}$ & $\mathbf{9.9\cdot10^{-3}}$ \\
        \hline
        \multirow{2}{*}{\footnotesize ReachLP-Unif} 
        & 4 & $0.212 \pm 0.002$ & $1.5\cdot10^{-1}$ \\
        & 16 & $3.149 \pm 0.004$ & $1.0\cdot10^{-2}$ \\
        \hline
        \multirow{2}{*}{\footnotesize ReachLP-GSG} 
        &  55 & $0.913 \pm 0.031$ & $5.3\cdot10^{-1}$\\
        & 205 & $2.164 \pm 0.042$ & $8.8\cdot10^{-2}$\\
        \hline
        \multirow{2}{*}{\footnotesize ReachLipBnB} 
        & $0.1$ & $0.956 \pm 0.067$ & $5.4\cdot10^{-1}$ \\
        & $0.001$ & $3.681 \pm 0.100$ & $1.2\cdot10^{-2}$ \\
        \hline
    \end{tabular}
    \caption{Performance of Theorem~\ref{thm:jacobian-based} for the discrete-Time double integrator model.}%
    \label{tab:DI_table}
\end{table}

\paragraph*{Discussion} Figure~\ref{fig:doubleintegrator} and Table~\ref{tab:DI_table} show that, when the open-loop system is linear, our interaction-based approach (Corollary~\ref{thm:linear}\ref{p2:linear}) combined with a suitable partitioning scheme beats state-of-the-art approaches in both accuracy and runtime. 

\subsection{ARCH-COMP Benchmarks}\label{sec:ARCH}

In the third experiment, we analyze three of the benchmarks from ARCH-COMP22~\cite{ARCH22:ARCH_COMP22_Category_Report_Artificial}. 

\paragraph*{Adaptive Cruise Control (ACC)}

We consider the Adaptive Cruise Control benchmark %
from~\cite[Equation (1)]{ARCH22:ARCH_COMP22_Category_Report_Artificial}. The neural network controller ($5\times 20\times 20\times 20\times 20\times 20\times 1$, ReLU activations) with the input $(x_{\mathrm{lead}}- x_{\mathrm{ego}}, v_{\mathrm{lead}}- v_{\mathrm{ego}}, a_{\mathrm{lead}}- a_{\mathrm{ego}}, v_{\mathrm{set}}, T_{\mathrm{gap}})^{\top}\in \real^5$ is applied with a zero-order holding of $0.1$ seconds~\cite{ARCH22:ARCH_COMP22_Category_Report_Artificial}. Our goal is to verify that from the initial set $\calX_0=[90,110]\times[32,32.2]\times[0,0]\times[10,11]\times[30,30.2]\times[0,0]$, the  collision specification
\begin{align*}
    x_{\mathrm{lead}} - x_{\mathrm{ego}} \ge  D_{\mathrm{default}} + T_{\mathrm{gap}} x_{\mathrm{ego}}
\end{align*}
is never violated in the next $5$ seconds, where $D_{\mathrm{default}}=10$ and $T_{\mathrm{gap}} = 1.4s$. To verify the specification, we define $D_{\text{rel}} = x_{\mathrm{lead}} - x_{\mathrm{ego}}$ and $D_{\text{safe}} = D_{\mathrm{default}} + T_{\mathrm{gap}} x_{\mathrm{ego}}$ and ensure that $ D_{\text{rel}} - D_{\text{safe}} \ge 0$ for the given time horizon.

We consider the interconnection-based approach with the closed-loop inclusion function $\OF^{\mathrm{con}}$ defined in~\eqref{eq:int-inclusion} constructed from the natural inclusion function for open-loop system %
and an affine inclusion function for the neural network. %
We use this interconnection-based inclusion function and Euler integration with a step-size of $0.01$ to compute upper and lower bounds on the states of the system starting from the initial set $\mathcal{X}_0$. These bounds are then used to obtain upper and lower bounds on $D_{\text{rel}}$ and $D_{\text{safe}}$. The results are shown in Figure~\ref{fig:ACC} \change{and the comparison between the runtime of our approach with POLAR~\cite{CH-JF-XC-WL-QZ:22} and JuliaReach~\cite{SB-etal:19} is provided in Table~\ref{tab:ARCH_table}.}%

\begin{figure}
    \centering
    \includegraphics[width=0.7\columnwidth]{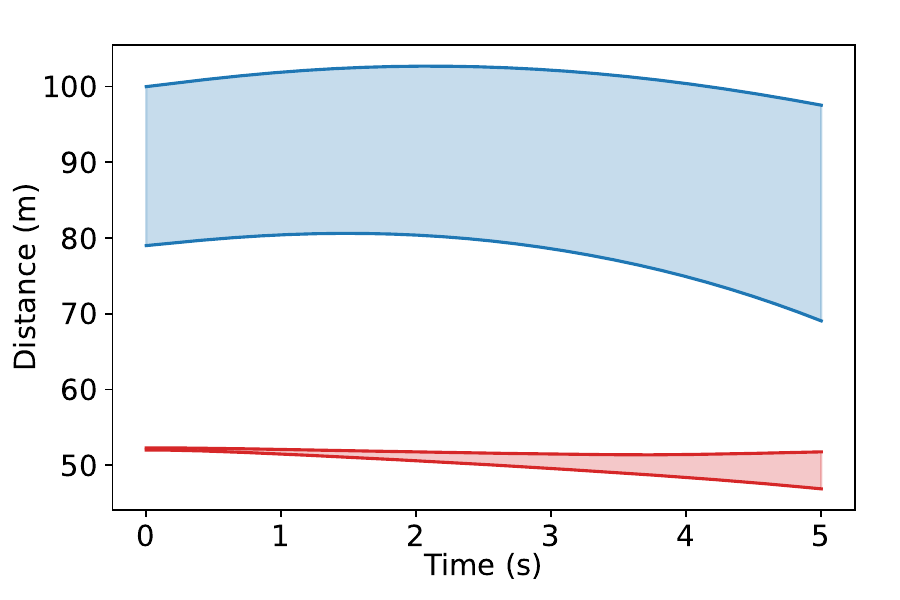}
    \caption{The interval estimates of the safety specifications for the Adaptive Cruise Control benchmark obtained from the inclusion function $\OF^{\mathrm{con}}$. The upper bound and lower bound on $D_\text{rel}$ are shown in {\color{tab:blue} blue} and the upper and lower bounds on $D_\text{safe}$ are shown in {\color{tab:red}red}. Our approach ensures that $D_\text{red}\geq D_\text{safe}$ for $t\in [0,5]$.}
    \label{fig:ACC}
\end{figure}

\paragraph*{2D Spacecraft Docking}

We consider the 2D spacecraft docking model \cite[Equation (13)]{ARCH22:ARCH_COMP22_Category_Report_Artificial}. The neural network controller $(4\times 4\times 256\times 256\times 4\times 2$, $\tanh$ activations) with a zero-order holding of $1$ second~\cite{ARCH22:ARCH_COMP22_Category_Report_Artificial}.  The goal is to verify that given an initial state set, the  safety specification
\begin{align} \label{eq:dockingspec}
   (\dot{s}^2_x + \dot{s}^2_x)^{\frac{1}{2}}\le 0.2 + 2n(s^2_x + s_y^2)^{\frac{1}{2}} 
\end{align}
is never violated in the next $40$ seconds. We define $H_L = (\dot{s}^2_x + \dot{s}^2_x)^{\frac{1}{2}}$ and $H_R = 0.2 + 2n(s^2_x + s_y^2)^{\frac{1}{2}}$ and our goal is to check that $H_L\le H_R$ for the next $40$ seconds. 

We consider the interconnection-based approach with the closed-loop inclusion function $\OF^{\mathrm{con}}$ defined in~\eqref{eq:int-inclusion} constructed from the natural inclusion function for the open-loop system %
and an affine inclusion function for the neural network. %
We use the closed-loop inclusion function $\OF^{\mathrm{con}}$ and Euler integration with a step-size $0.1$ to obtain upper and lower bounds on the states of the system starting from four different initial sets $\calX_0^i$ for $i\in\{0,1,2,3\}$ as shown in Figure~\ref{fig:docking}. These bounds are then used to provide upper and lower bounds on $H_L$ and $H_R$. The results are shown in Figure~\ref{fig:docking} \change{and the comparison between the runtime of our approach with POLAR~\cite{CH-JF-XC-WL-QZ:22} and JuliaReach~\cite{SB-etal:19} is provided in Table~\ref{tab:ARCH_table}.} %

\begin{figure}
    \centering
    \includegraphics[width=1\columnwidth]{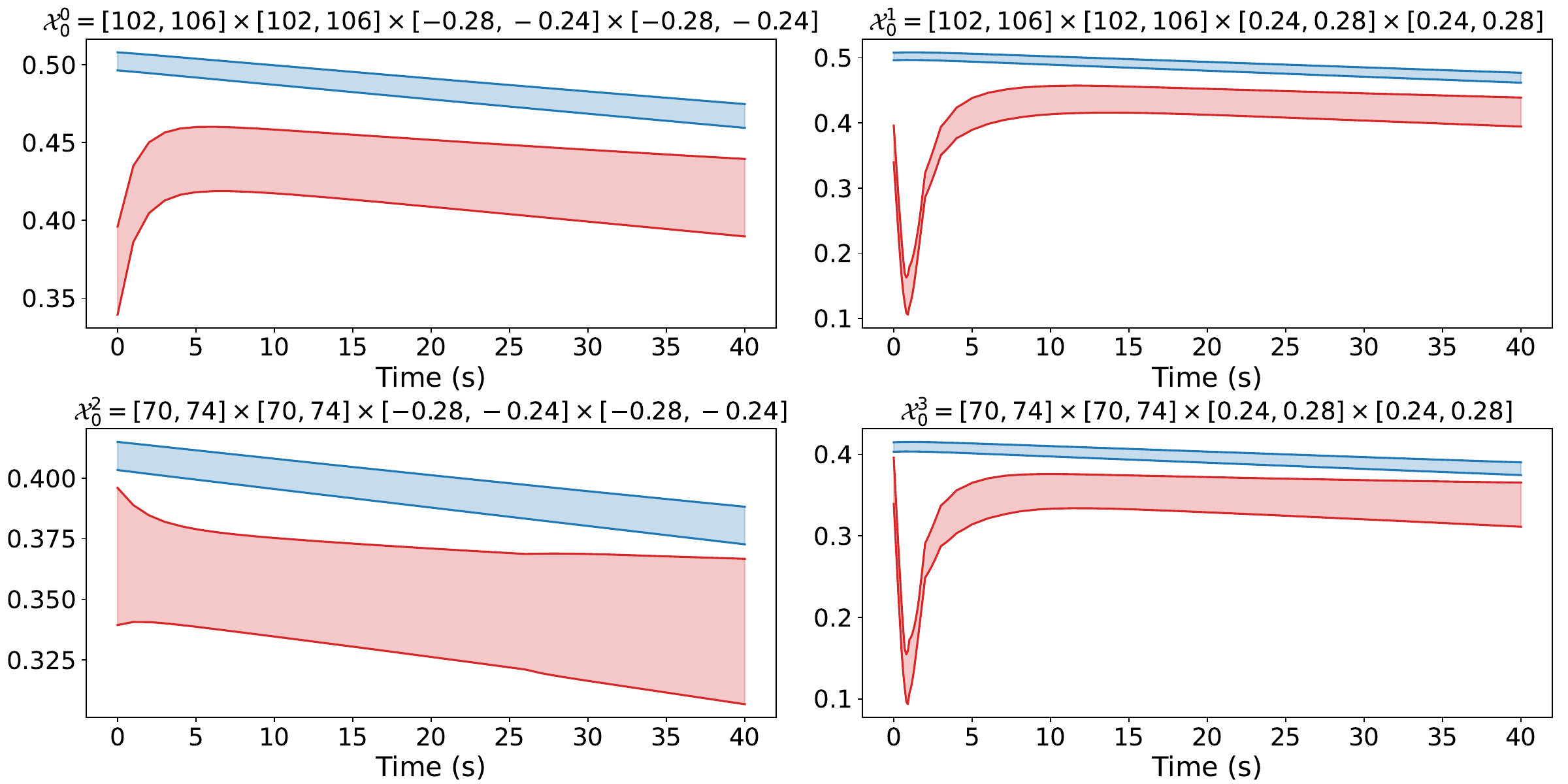}
    \vspace{-0.5cm}
    \caption{The interval estimates of the safety specification for the 2D Spacecraft Docking benchmark obtained from the interconnection-based inclusion function $\OF^{\mathrm{con}}$ are pictured for four different initial sets $\{\calX_0^i\}_{i=0}^3$. The upper and lower bounds on $H_{L}$ are shown in {\color{tab:red}red} and the upper and lower bounds on $H_R$ are shown in {\color{tab:blue}blue}. Our approach verifies that $H_L\le H_R$ for $t\in [0,40]$.}
    \label{fig:docking}
\end{figure}

\paragraph*{Sherlock-Benchmark-9 (TORA)}
We consider the translational oscillations by a rotational actuator (TORA) benchmark from~\cite[Equation (2)]{ARCH22:ARCH_COMP22_Category_Report_Artificial}.  The neural network controller ($4\times 20\times 20 \times 20 \times 1$ with $\mathrm{ReLU}$ activation functions and an additional layer with $\mathrm{tanh}$ activation function) is applied with a zero-order holding of $0.5$ seconds~\cite{ARCH22:ARCH_COMP22_Category_Report_Artificial}. The goal is to verify whether the system reaches the region $[-0.1,0.2]\times [-0.9,-0.6]\times \real^2$ from the initial set $\mathcal{X}_0=[-0.77,-0.75]\times [-0.45,-0.43]\times[0.51,0.54]\times[-0.3,-0.28]$ within  $5$ seconds.

We add the affine redundant variables $y_1= x_1+x_2$ and $y_2=x_1-x_2$ and analyze the augmented system in $z = (x_1,x_2,y_1,y_2)^{\top}$. 
We consider the interaction-based approach with the mixed Jacobian-based cornered inclusion function $\OF^{\mathrm{act}}$ as defined in~\eqref{eq:jac-inclusion} constructed for four corners $(z,u),(z,\widehat{u}),(\widehat{z},u),(\widehat{z},\widehat{u})$ (cf. Remark~\ref{rem:important}). %
We use the intersection of the two inclusion functions $\OF^{\mathrm{con}}\wedge \OF^{\mathrm{act}}$ with Euler integration step size of $0.005$ to provide over-approximations of reachable sets of the system starting from $\mathcal{X}_0$. The results are shown in Figure~\ref{fig:TORA} \change{and the comparison between the runtime of our approach with POLAR~\cite{CH-JF-XC-WL-QZ:22} and JuliaReach~\cite{SB-etal:19} is provided in Table~\ref{tab:ARCH_table}.} %

\begin{figure}
    \centering
    \includegraphics[width=\columnwidth]{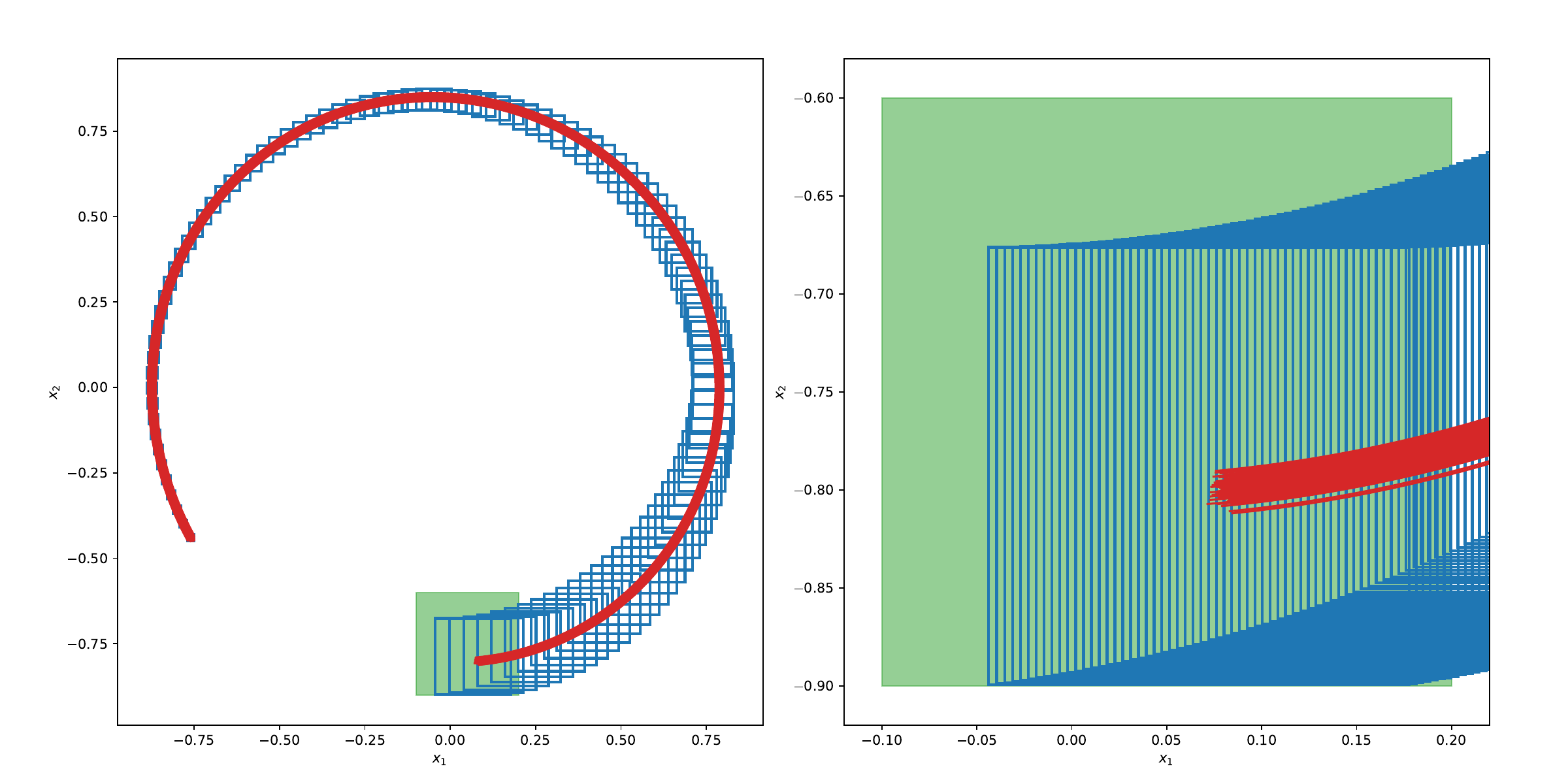}
    \vspace{-0.5cm}
    \caption{The reachable set for the TORA benchmark is shown on the $x_1$-$x_2$ plane. The goal set is pictured in {\color{tab:green}green}, and 100 true trajectories of the system are pictured in {\color{red}red}. \textbf{Left:} Every tenth over-approximation of the reachable set is pictured in {\color{tab:blue}blue}. \textbf{Right:} The plot is zoomed into the goal set, and all over-approximations of the reachable set are shown in {\color{tab:blue} blue}.
    \change{We numerically verify that the final reachable set is fully contained inside of the goal region.}
    }
    \label{fig:TORA}
\end{figure}

\begin{table}[h!]
    \centering
    \begin{tabular}{||c c c c||}
        \hline
        System & ReachMM & POLAR & JuliaReach \\
        \hline\hline
        ACC & $0.398\pm 0.010$ & $1.388 \pm 0.022$ & $0.255$ \\
        Docking & $0.243\pm 0.004$ & $43.696 \pm 0.141$ & N/A\tablefootnote{JuliaReach did not attempt this specification.} \\
        TORA & $5.586\pm0.029$ & $0.139\pm0.008$ & $0.690$ \\
        \hline
    \end{tabular}
    \caption{Summary of runtimes (s) from the ARCH-COMP benchmarks.}
    \label{tab:ARCH_table}
\end{table}

\paragraph*{Discussion} We showed that for verification of safety specification on a selected number of examples from ARCH-COMP~\cite{ARCH22:ARCH_COMP22_Category_Report_Artificial}, the runtime of our approach is comparable with the state-of-the-art methods POLAR~\cite{CH-JF-XC-WL-QZ:22} and JuliaReach~\cite{CS-MF-SG:22}. On the other benchmarks in~\cite{ARCH22:ARCH_COMP22_Category_Report_Artificial}, our method is unable to verify the desired specification \change{at this time, due to their large time horizons and excessive over-approximation buildup.} In the next section, \change{however,} we show that our interconnection-based and interaction-based methods enjoy better scalability compared to POLAR and JuliaReach.

\subsection{Vehicle Platooning} \label{subsec:platooning}

Finally, we investigate the scalability of our method. %
We consider a platoon of $N$ vehicles $\calV = \{\calV_j\}_{j=1}^N$, each with the two-dimensional double integrator dynamics
\begin{xalignat}{2} \label{eq:platoondyn}
    \dot{p}_x^j &= v_x^j, & \dot{v}_x^j &= \sigma(u_x^j) + w_x^j, \nonumber\\
    \dot{p}_y^j &= v_y^j, & \dot{v}_y^j &= \sigma(u_y^j) + w_y^j,
\end{xalignat} 
where $p^j = (p_x^j,p_y^j)\in\R^2$ is the displacement of the center of mass of $\calV_j$, $v^j = (v_x^j,v_y^j)\in\R^2$ is the velocity of the center of mass of $\calV_j$, $(u_x^j,u_y^j)\in\R^2$ are desired acceleration inputs limited by the softmax operator $\sigma(u)=u_{\text{lim}}\tanh(u/u_{\text{lim}})$ with $u_\text{lim}=5$, and $w_x^j,w_y^j\sim\calU([-0.001,0.001])$ are uniformly distributed disturbances. We consider a leader-follower structure for the system, where the first vehicle $\calV_1$ chooses its control $u=(u_x^1,u_y^1)$ as the output of a neural network ($4\times 100 \times 100 \times 2$, ReLU activations), and the rest of the platoon $\{\calV_j\}_{j=2}^N$ applies a PD tracking control input
\begin{align}
    u_{\bfd}^j = k_{p}\left(p_\bfd^{j-1} - p_\bfd^{j} - r\frac{v_\bfd^{j-1}}{\|v^{j-1}\|_2}\right) + k_{v}(v_\bfd^{j-1} - v_\bfd^j),
\end{align}
for each $\bfd\in\{x,y\}$ with $k_p=k_v=5$ and $r=0.5$. The neural network was trained using data from an offline MPC policy for the leader only ($N=1$). The offline policy minimized a quadratic cost stabilizing to the origin while avoiding a circular obstacle centered at $(4,4)$ with radius $2.25$, implemented as a hard constraint with $33\%$ padding and a slack variable.

\begin{table}[]
    \centering
    \begin{tabular}{||c|c|c|c|c|c||}
        \hline
        $N$ (units) & $\#$ of states & $\OF^{\mathrm{con}}$ & $\OF^{\mathrm{act}}$ & POLAR & JuliaReach \\
        \hline\hline
         1           & 4      & $0.297$ & $0.635$  & $9.352$ & $0.224$ \\
         4           & 16     & $0.399$ & $1.369$  & $14.182$ & $12.579$ \\
         9           & 36     & $0.574$ & $3.144$  & $43.428$ & $59.929$ \\
         20          & 80     & $0.999$ & $9.737$  & $316.337$ & $-$\tablefootnote{\label{foot:juliareach}No matter the choice of hyper-parameter, the package failed, either 1) citing a maximum number of validation steps reached or 2) timing out.} \\
         50          & 200    & $2.420$ & $46.426$ & $4256.435$ & $-$ \\
         \hline
    \end{tabular}
    \caption{The runtimes (s) for the platooning benchmark.}
    \label{tab:platooning}
    \vspace{-0.5cm}
\end{table}
We consider the interconnection-based approach with closed-loop inclusion function $\OF^{\mathrm{con}}$ defined in~\eqref{eq:int-inclusion} constructed from the natural inclusion function for open-loop system. %
We also consider the interaction-based approach with the closed-loop mixed Jacobian-based cornered inclusion function $\OF^{\mathrm{act}}$ as defined in~\eqref{eq:jac-inclusion} constructed for  four corners $(x,u),(u,\widehat{u}),(\widehat{x},u),(\widehat{x},\widehat{u})$ (cf. Remark~\ref{rem:important}).  
We perform reachability analysis for platoons of $N$ vehicles with $N\in \{1,4,9,20,50\}$. %
For $j$th vehicle in the platoon, we compute the reachable sets for the time frame $t\in [0,1.5]$ using the closed-loop inclusion functions $\OF^{\mathrm{con}}$ and $\OF^{\mathrm{act}}$ starting from the initial set $ ([7.225+0.5(j-1)\cos(\pi/3),7.275+0.5(j-1)\cos(\pi/3)] \times [5.725+0.5(j-1)\sin(\pi/3),5.775+0.5(j-1)\sin(\pi/3)]\times[-0.5,-0.5]\times[-5,-5])$. All integration is performed using \change{Euler integration} with a step size of $0.0125$. Figure~\ref{fig:platooning} visualizes the result for $N=9$. Runtimes  of our approach as well as POLAR~\cite{CH-JF-XC-WL-QZ:22} and JuliaReach~\cite{CS-MF-SG:22} for platoons of varying size are reported in Table~\ref{tab:platooning}. Note that the implementations in POLAR and JuliaReach omit the disturbances $w^j_{\mathbf{d}}$ and the softmax $\sigma$ in equation~\eqref{eq:platoondyn}.

%

\begin{figure}
    \includegraphics[width=0.495\columnwidth]{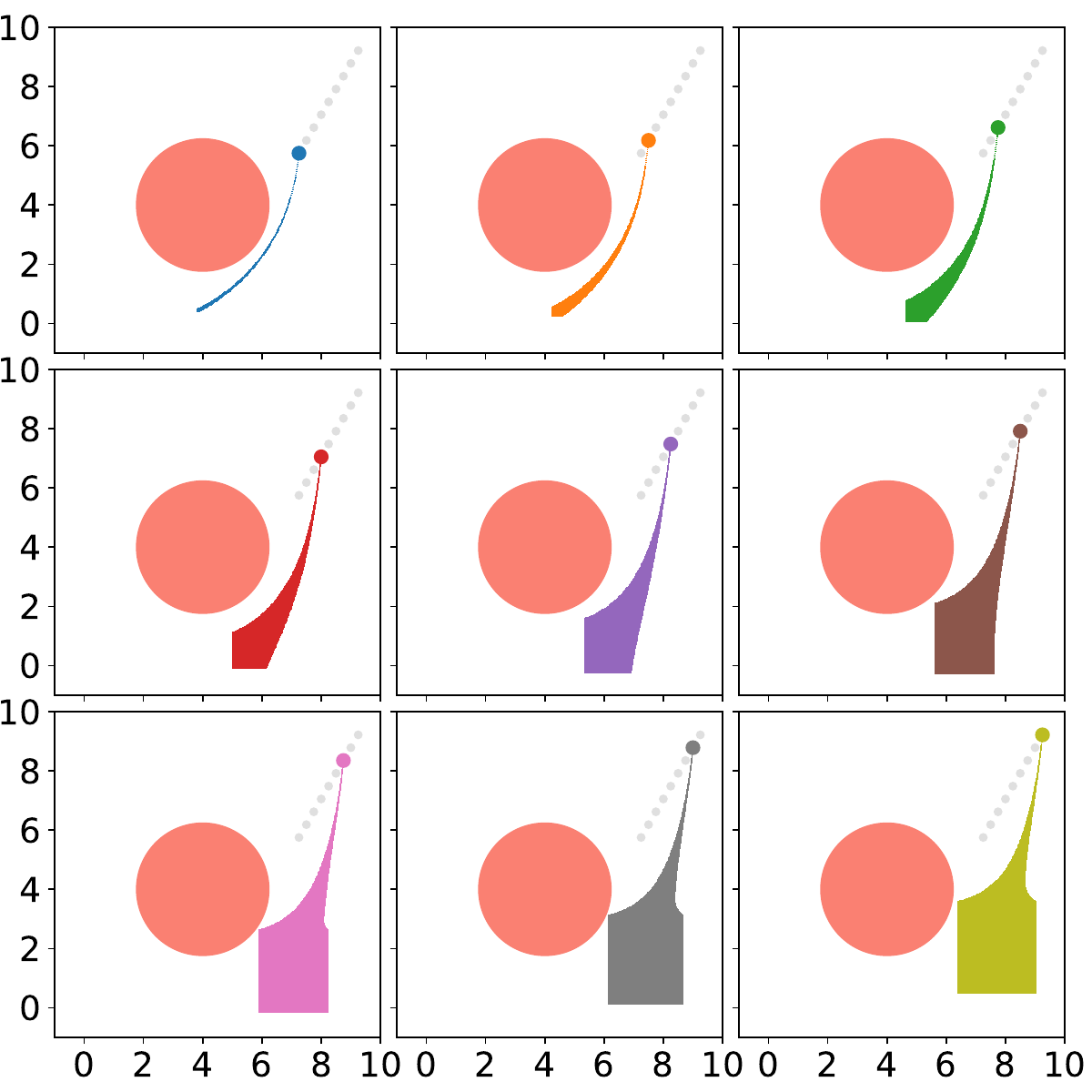}
    \includegraphics[width=0.495\columnwidth]{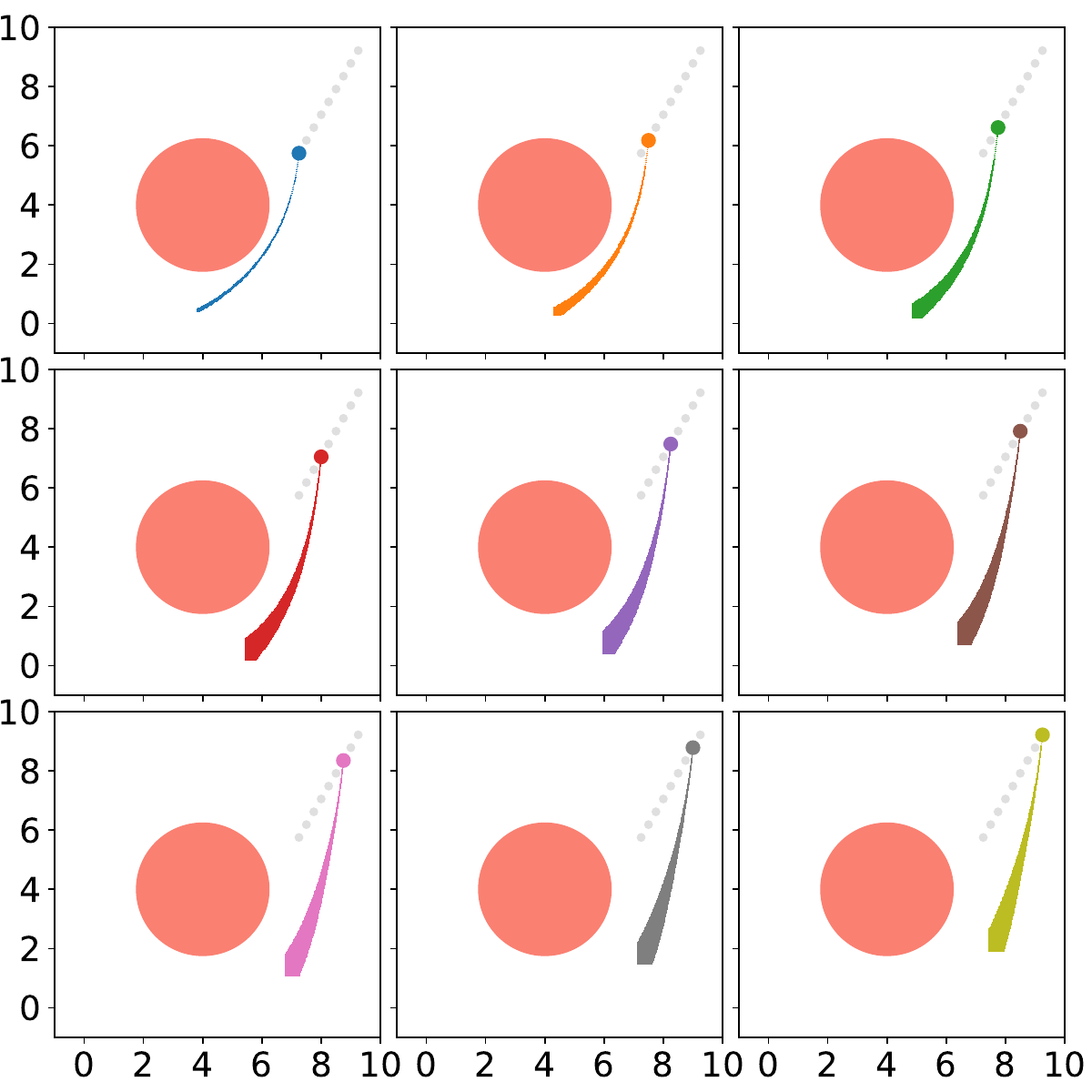}
    \vspace{-0.5cm}
    \caption{ \textbf{Left:} The over-approximation of the reachable sets computed using $\OF^{\mathrm{con}}$ for each individual unit in the platooning example with $N=9$ are shown on $p_x$-$p_y$ planes, with the circular obstacle in {\color{salmon}salmon}. 
    \textbf{Right:} The over-approximation of the reachable sets computed using $\OF^{\mathrm{act}}$ for each individual unit in the platooning example with $N=9$ are shown on $p_x$-$p_y$ planes, with the circular obstacle in {\color{salmon}salmon}. }
    \label{fig:platooning}
\end{figure}

\paragraph*{Discussion}
    This experiment demonstrates one of the key features of our interval analysis framework---its scalability to large-scale systems.  
    In general, Figure~\ref{fig:platooning} shows how $\OF^\text{act}$ outperforms  $\OF^{\text{con}}$ in accuracy of the reachable set over-approximation. However, Table~\ref{tab:platooning} shows how $\OF^\text{con}$ outperforms $\OF^\text{act}$ in terms of scalability with respect to state dimensions. Moreover, both $\OF^\text{con}$ and $\OF^\text{act}$ demonstrate better runtime scalability than POLAR and JuliaReach.

\section{Conclusions}

We present a framework based on interval analysis for safety verification of neural network controlled systems. 
The main idea is to embed the closed-loop system into a higher dimensional system whose over-approximation of reachable sets can be easily computed using its extreme trajectory. 
Using an inclusion function of the open-loop system and interval bounds for neural networks, we proposed an interconnection-based approach and an interaction-based approach for constructing the closed-loop embedding system. 
We show how these approaches can capture the interactions between the
nonlinear plant and the neural network controllers and we provide
several numerical experiments comparing them with the state-of-the-art
algorithms in the literature.

\section*{References}
\vspace{-0.5cm}
\bibliographystyle{IEEEtran}
\bibliography{SJ,SJ2}

\end{document}

%% file: macros.tex
\newcommand{\R}{\mathbb{R}}

\newcommand{\calT}{\mathcal{T}}
\newcommand{\calU}{\mathcal{U}}
\newcommand{\calV}{\mathcal{V}}

\newcommand{\calX}{\mathcal{X}}

\newcommand{\sfG}{\mathsf{G}}

\newcommand{\bfd}{\mathbf{d}}

\newcommand{\ul}[1]{\underline{#1}}

\newcommand{\ulw}{\ul{w}}

\newcommand{\ulM}{\ul{M}}

\newcommand{\ol}[1]{\overline{#1}}

\newcommand{\olw}{\ol{w}}

\newcommand{\olM}{\ol{M}}

\definecolor{tab:blue}{HTML}{1f77b4}
\definecolor{tab:orange}{HTML}{ff7f0e}
\definecolor{tab:green}{HTML}{2ca02c}
\definecolor{tab:red}{HTML}{d62728}
\definecolor{tab:purple}{HTML}{9467bd}
\definecolor{tab:brown}{HTML}{8c564b}
\definecolor{tab:pink}{HTML}{e377c2}
\definecolor{tab:gray}{HTML}{7f7f7f}
\definecolor{tab:olive}{HTML}{bcbd22}
\definecolor{tab:cyan}{HTML}{17becf}
\definecolor{salmon}{HTML}{fa8072}